\documentclass[10pt, twocolumn, conference]{IEEEtran}

\makeatletter
\def\ps@headings{%
\def\@oddhead{\mbox{}\scriptsize\rightmark \hfil \thepage}%
\def\@evenhead{\scriptsize\thepage \hfil \leftmark\mbox{}}%
\def\@oddfoot{}%
\def\@evenfoot{}}
\makeatother
\usepackage{amsmath, mathrsfs,amsfonts}
\usepackage{amssymb}
\usepackage{graphicx}
\usepackage{cite}
\usepackage{algorithmic}
\usepackage{algorithm2e}
\usepackage{array}
\usepackage{subfigure,enumerate}

\def\ignore#1{}

\newtheorem{definition}{Definition}

\newtheorem{theorem}{Theorem}
\newtheorem{lemma}{Lemma}
\newtheorem{corollary}{Corollary}
\newtheorem{proposition}{Proposition}

               {\begin{list}{}{\leftmargin#1\rightmargin#2\topsep#3}\item{}}%
               {\end{list}}

\usepackage[top=0.75in, bottom=0.75in, left=0.5in, right=0.5in]{geometry}

\IEEEoverridecommandlockouts
\begin{document}
\title{Robustness of Bidirectional Interdependent Networks: Analysis and Design}

\author{\IEEEauthorblockN{Marzieh Parandehgheibi and Eytan Modiano }
\IEEEauthorblockA{Laboratory for Information and Decision Systems, Massachusetts Institute of Technology, Cambridge, MA, USA\\}}

\maketitle 

\begin{abstract}
We study the robustness of interdependent networks where two networks are said to be interdependent if the operation of one network depends on the operation of the other one, and vice versa. In this paper, we propose a model for analyzing bidirectional interdependent networks with known topology. We define the metric $\mathcal{MR}(D)$ to be the minimum number of nodes that should be removed from one network to cause the failure of $D$ nodes in the other network due to cascading failures. We prove that evaluating this metric is not only NP-complete, but also inapproximable. Next, we propose heuristics for evaluating this metric and compare their performances using simulation results. Finally, we introduce two closely related definitions for robust design of interdependent networks; propose algorithms for explicit design, and demonstrate the relation between robust interdependent networks and expander graphs.
\end{abstract}

\IEEEpeerreviewmaketitle

\section{Introduction}\label{Intro_sec}

Many of today's infrastructures are organized in the form of networks and are becoming increasingly interdependent. For example, the power grid and communication networks have a cyber-physical interdependency where the power nodes depend on the control signals coming from the communication nodes and communication nodes operate using the power coming from the power nodes. As another example, the power grid and gas networks have a physical-physical interdependency where the compressors in gas networks require power from the power grid to transmit gas and the gas generators in the power grid require gas to generate power. 

Although interdependency is required for the operation of both networks under normal conditions, if a failure happens inside one of the networks it can cascade to the other network. For example, if a failure occurs inside the power grid, some of the communication nodes will lose their power and fail. As a result, new power nodes lose their control and fail which can lead to the failure of additional communication nodes. Thus, a cascade of failures can occur between the two networks due to the strong interdependency.

The concept of interdependency was first introduced by Rinaldi \textit{et. al.} in \cite{rinaldi2001identifying}, where the authors described different types of interdependencies as well as different types of failures that can occur in interdependent systems. In 2010, Buldyrev \textit{et. al.} developed the first mathematical model for describing interdependency between two networks \cite{buldyrev2010catastrophic}. They modeled each network as a random graph and assumed there exists a one-to-one interdependency between the two networks. They consider two networks $A$ and $B$, where node in network $A$ can operate if it is connected to the largest connected component in network $A$ and its correspondent node in network $B$ is also operating. Using percolation theory, they showed that failures spread more in interdependent networks than isolated networks; thus, interdependent networks are more vulnerable. Using this model, Parshani \textit{et al.} showed that reducing the dependency between the networks makes them more robust to random failures \cite{parshani2010interdependent}. In \cite{buldyrev2011interdependent}, the robustness of a slightly different version of interdependent networks was investigated where mutually dependent nodes have the same number of connectivity links. A discussion of follow-up works on Buldyrev's model can be found in \cite{gao2012networks}. 

The robustness of interdependent networks under targeted attacks was studied in \cite{huang2011robustness}. The stability of interdependent spatially embedded networks, and the impact of geographical attacks on the robustness of two interdependent spatially embedded networks was studied in \cite{bashan2013extreme} and \cite{berezin2013spatially}, respectively. Moreover, in \cite{Sundaram2015:Online}, the notion of interdependency was generalized to more than two networks, and the ability of networks to tolerate certain structural attacks was investigated.

There have also been some efforts on designing robust interdependent networks. The authors in \cite{schneider2013towards} proposed a strategy based on ``betweenness" centrality measurement to make a minimum number of nodes resilient such that the overall robustness of networks is increased. In \cite{stippinger2013enhancing} a dynamic enhancing model was studied, where the authors defined a healing process, where an interdependency link is formed with some given probability. They showed that there is a threshold for this probability where catastrophic failures occur below the threshold and do not occur above the threshold. Finally, the authors in \cite{yagan2012optimal} showed that the robustness of interdependent networks depends on the allocation of the interdependency links, and characterized an optimum allocation against random attacks.

As described above, most of the literature on interdependency follows the model of \cite{buldyrev2010catastrophic}, and relies on the asymptotic behavior of random networks. However, in reality networks are not random, and models and analytical tools are needed for studying networks with known topology. There have been a few attempts on assessing deterministic interdependent networks in the literature. In \cite{parandehgheibi2013robustness}, we proposed a new connectivity model for deterministic networks with known topologies. In this model, every node has to be connected to a source node in the network. For example, in the case of the power grid, a load has to be connected to the generator as its source. The interdependency model was extended in \cite{sen2014identification} and \cite{Das2014Root} by introducing different types of interdependency, where a node depends on multiple types of nodes for its operation.

There have also been some attempts at modeling the interactions between networks in more realistic settings. The impact of failures in the power grid on the communication network was explored in \cite{rosato2008modelling}, and the impact of power failures on layered optical communication networks was studied in \cite{parandehgheibi2014survivable}. In \cite{parandehgheibi2014mitigating,korkali2014reducing}, the interdependency model for power grid and communication network was modified by considering the power flow equations and cascading failures inside the power grid. In particular, we showed in \cite{parandehgheibi2014mitigating} that under some conditions, the connectivity model proposed in this paper can be used to describe the cascading failures between interdependent power grid and communication networks.

In this paper, we propose a new model for interdependent networks with ``known'' topologies. We propose several metrics for identifying the impact of failures in one network on the vulnerability of the other one due to interdependency, analyze the complexity of our metrics, and propose algorithms for approximating them. 

Next, we prove that interdependent networks with bidirectional edges are more robust than those with unidirectional edges, and propose two closely-related definitions for robust interdependent networks: (1) a lexicographic definition which guarantees that networks robust to large failures are also robust to small failures, and (2) a relative definition which guarantees that the ratio of the size of the initial failure to the size of the final failure is large.
	
Finally, we propose explicit algorithms for robust design of networks under the first definition and showed the relation of robust interdependent networks with expander graphs under the second definition.

The rest of this paper is organized as follows. In Section \ref{Model_sec}, we introduce our model for interdependent networks. In Section \ref{Analysis_sec}, we formulate two closely-related metrics for vulnerability assessment, analyze their complexity and propose heuristics for approximating them. Next, in Section \ref{Design_sec}, we introduce two definitions for robust networks and propose algorithms for allocating the interdependency links in order to obtain the most robust bidirectional interdependent networks. Finally, in Section \ref{Discussion_sec}, we discuss the robustness of interdependent networks with general topologies, and conclude in Section \ref{Conclusion_sec}.

%%%%%%%%%%%%%%%%%%%%%%%%%%%%%%%%%%%%%%%%%%%%%%%%%%%%%%%%%%%%%%%%%%%%%%%%%%%%%%%%%%%%%%%%%%%%
\section{Interdependency Model}\label{Model_sec}

Consider network $A$ with $N_1$ nodes and a set of robust source nodes $S_A$\footnote{We assume that source nodes do not fail.} where every node in $A$ is connected to at least one source node in $S_A$ via a path in network $A$. Similarly, network $B$ has $N_2$ nodes and a set of robust source nodes $S_B$ where every node in $B$ is connected to at least one source node in $S_B$ via a path in $B$. Without loss of generality, we assume that each network has exactly one source node, where one can replace all source nodes in $S_A$ ($S_B$) with one dummy node called $S_A$ ($S_B$). In addition, there exists an interdependency between nodes in networks $A$ and $B$ as follows: 
\begin{enumerate}
\item Every node in network $A$ receives at least one incoming edge from a node in network $B$; 
\item Every node in network $B$ receives at least one incoming edge from a node in network $A$; 
\end{enumerate}

See Figure \ref{Model} for an example of our interdependent network.

\begin{figure}[ht]
\centering
\includegraphics[scale=0.33, angle=-90]{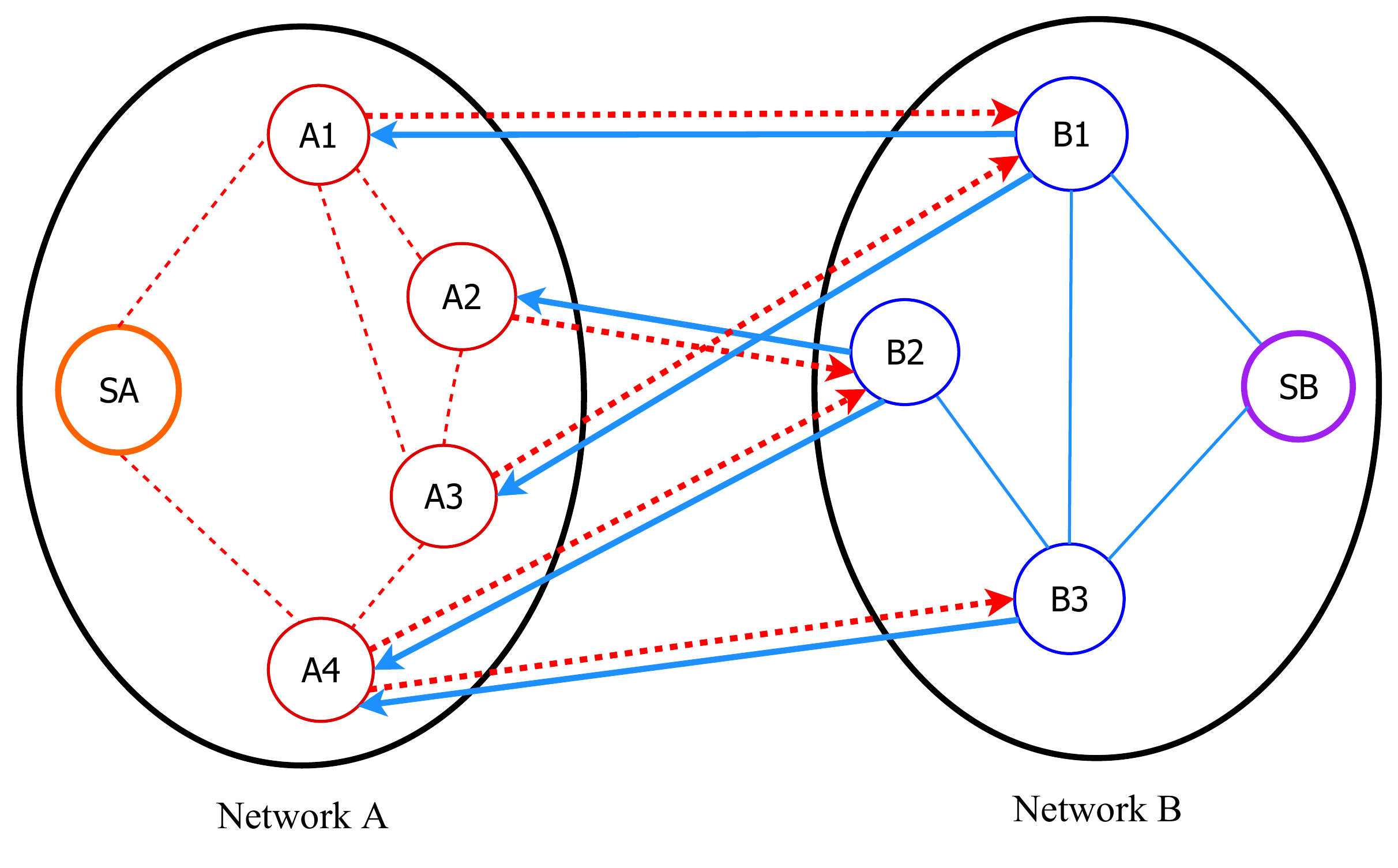}\vspace{-0.3cm}
\caption{Interdependency Model - Dotted lines represent links of type A and solid lines represent links of type B.}
\label{Model}\vspace{-0.2cm}
\end{figure}

According to our model, a node in network $A$ is operating if (a1) it is connected to source $S_A$ via a path of operating nodes in $A$, and (a2) it is connected to at least one operating node in network $B$. Similarly, a node $B$ is operating if (b1) it is connected to source $S_B$ via a path of operating nodes in $B$, and (b2) it is connected to at least one operating node in network $A$. Note that condition (a2) guarantees the connection of node $A$ to source $S_B$, as well. This is due to the fact that an operating node of type $B$ should be connected to $S_B$. Similarly, condition (b2) guarantees the connection of node $B$ to source $S_A$.

It is worthwhile to note a critical difference between the interdependent networks and isolated networks which makes the analysis of interdependent networks more complex. In our model, every node of type $A$ will be operational if it is receiving incoming edges from both type $A$ and type $B$ nodes; however, its outgoing edge will be ``only" of type $A$. Therefore, although each operational node is connected to both sources $S_A$ and $S_B$ via two paths, the type of nodes in each path also matters.

\begin{figure}[h]
\centering
\subfigure[Step1 - $A_4$ fails, initially]
{\label{step1}\includegraphics[scale=0.29, angle=-90]{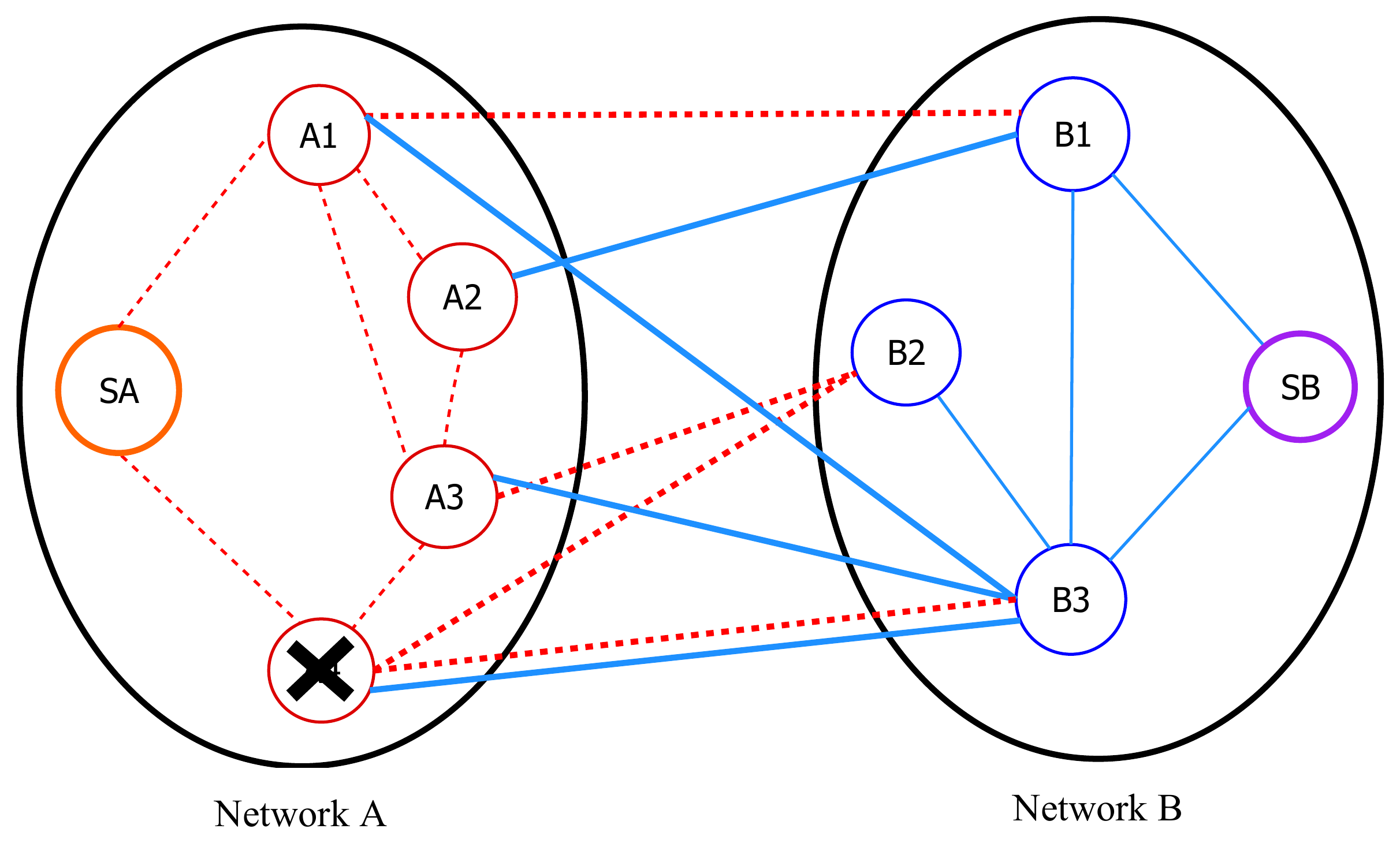}}                
\subfigure[Step2 - $B_3$ fails]
{\label{step2}\includegraphics[scale=0.29, angle=-90]{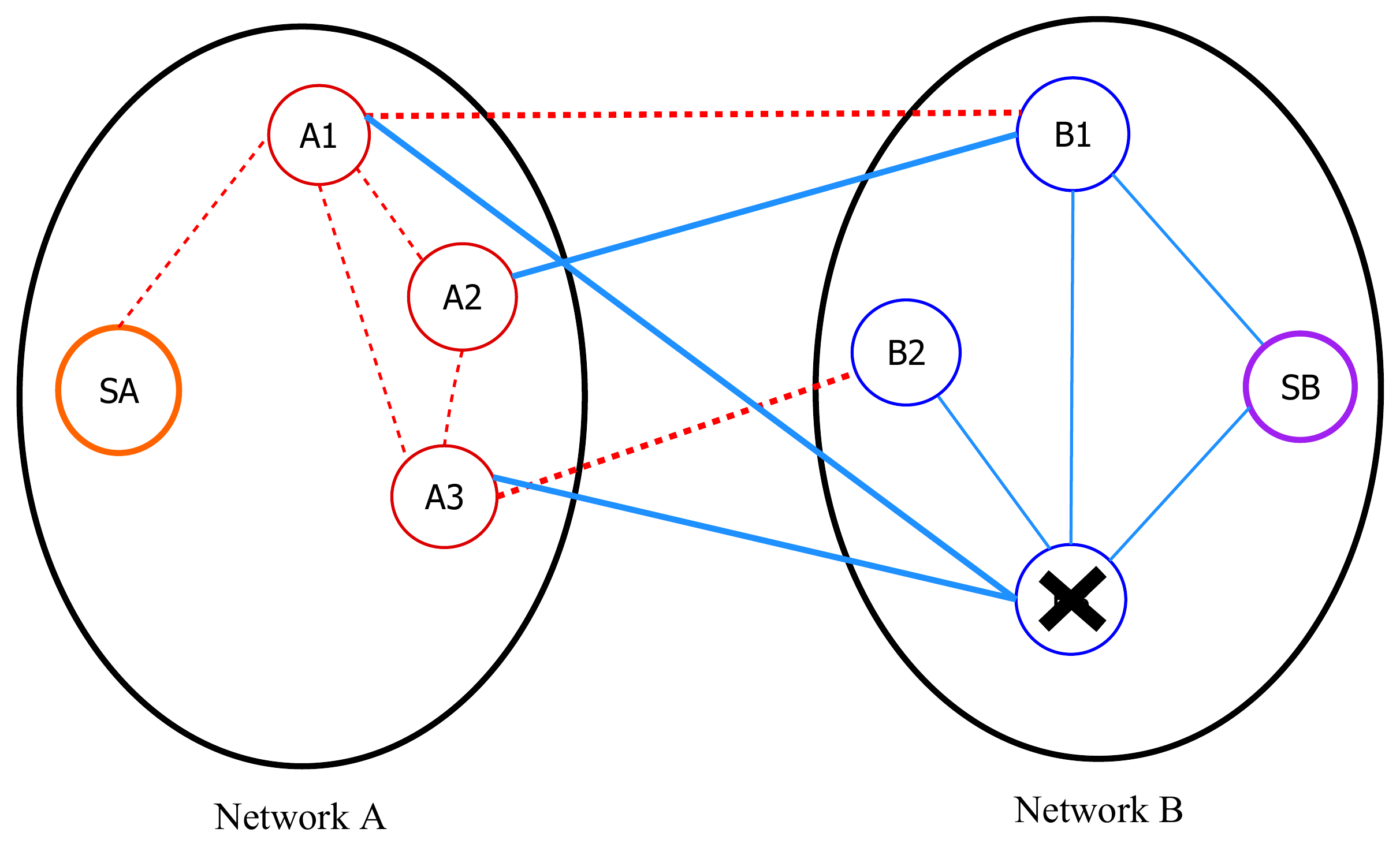}}
\subfigure[Step3 - $A_1$, $A_3$ and $B_2$ fail]
{\label{step3}\includegraphics[scale=0.29, angle=-90]{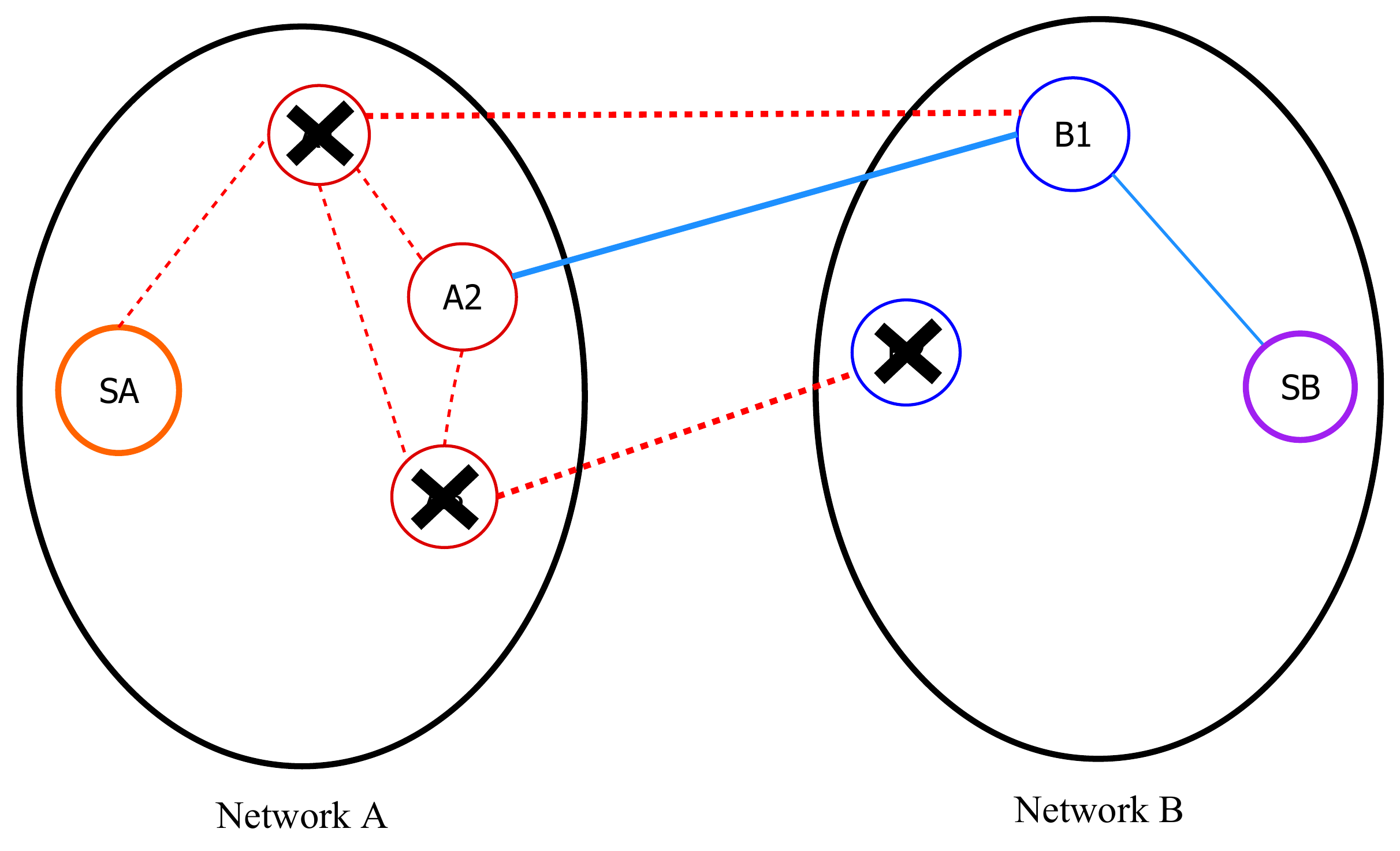}}
\subfigure[Step4 - $B_1$ and $A_2$ fail]
{\label{step4}\includegraphics[scale=0.29, angle=-90]{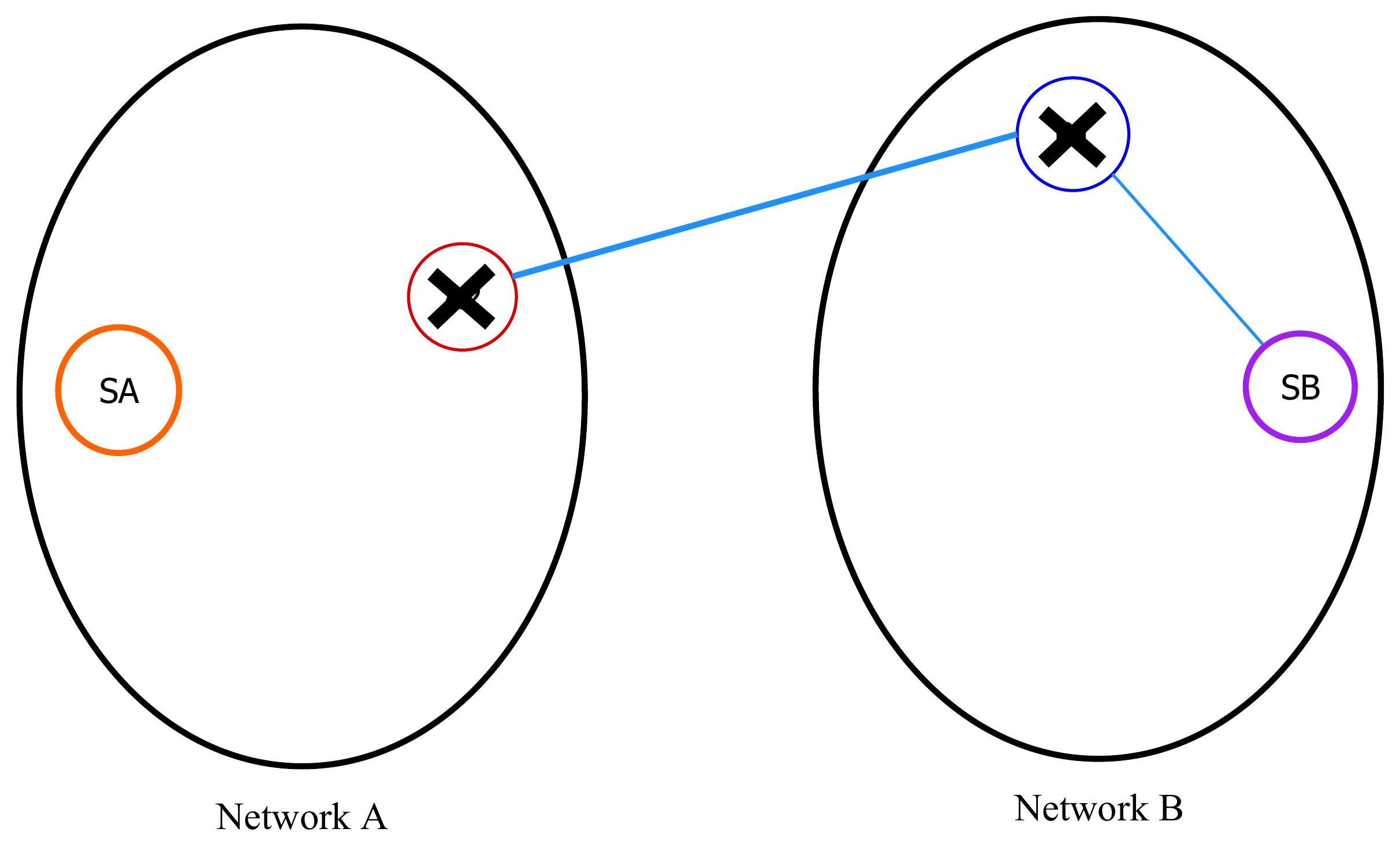}}
\caption{Cascade of a single failure in an interdependent model}
\label{Cascade}
\end{figure}

\subsection{Failure Cascades}

We start with an example demonstrating that a single failure can cascade multiple times within and between networks $A$ and $B$ (Figure \ref{Cascade}). Suppose that initially node $A_4$ fails (Step 1). As a result, all the edges attached to $A_4$ fail, and node $B_3$ loses its connection to network $A$ and fails (Step 2); Consequently, node $A_1$ and $A_3$ lose their connection to network $B$, and node $B_2$ loses its connection to source $S_B$, and all fail (Step 3). Finally node $B_1$ loses its connection to network $A$, and substation $A_2$ loses its connection to source $S_A$, and both fail (Step 4).

In this paper, the interdependent networks $A$ and $B$ have special \textit{star topologies}; i.e. every node is directly connected to the source in that network. In a star topology, failure of a node in network $A$ cannot disconnect other nodes in network $A$ from source $S_A$; and similarly, failure of a node in network $B$ cannot disconnect other nodes in network $B$ from source $S_B$. Therefore, any cascading failure in the system would be \textit{only} due to the interdependency between the networks. We consider this topology as it gives us the opportunity to investigate the impact of interdependency on the robustness of networks.%Then, we extend our analysis to networks with acyclic topologies; i.e. trees.

\subsection{Types of Interdependency}

One can consider both unidirectional and bidirectional interdependency. In unidirectional interdependency, interdependent edges are unidirectional; i.e. if node $i$ in network $A$ supports node $j$ in network $B$, it is not necessarily supported by node $j$. In bidirectional interdependency, interdependent edges are bidirectional; i.e. if node $i$ in network $A$ supports node $j$ in network $B$, it is also supported by node $j$.

The main difference between the cascade of failures in these two networks is the fact that in unidirectional networks, a failure can cascade in multiple stages, whereas in a bidirectional network, a failure cascades only in one stage \footnote{Suppose failure cascades from $i_1$ to $i_2$ to $i_3$; i.e. two stages. This means that node $i_2$ has two neighbors ; i.e. two incoming edges; thus, loss of neighbor $i_1$ does not lead to the failure of node $i_2$} (See \cite{parandehgheibi2013robustness} for more details). Later, we will see that the bidirectional interdependent networks are more robust than the unidirectional interdependent networks due to this difference.

%%%%%%%%%%%%%%%%%%%%%%%%%%%%%%%%%%%%%%%%%%%%%%%%%%%%%%%%%%%%%%%%%%%%%%%%%%%%%%%%%%%%%%%%%%%%%%%%%%%%%%%%%%%%%
\section{Analysis}\label{Analysis_sec}
In this section, we analyze the robustness of interdependent networks with known topology to cascading failures. First, we define two closely related metrics which find the most influential nodes in an interdependent network. Then, we formulate these metrics, investigate their complexity and propose algorithms for evaluating them.

\subsection{Metrics}

We define two metrics $\mathcal{MR}(D)$ and $\mathcal{MRB}(D)$ to evaluate the impact of cascading failures in interdependent networks.

\begin{definition}
In an interdependent network, metric $\mathcal{MR}(D)$ denotes the minimum number of node removals from network $A$ which causes the failure of $D$ \textit{arbitrary} nodes in network $B$.
\end{definition}

\begin{definition}
In an interdependent network, metric $\mathcal{MRB}(D)$ denotes the minimum number of node removals from \textit{both} networks which causes the failure of $D$ \textit{arbitrary} nodes in network $B$.
\end{definition}

Note that due to symmetry, one can define the same metrics for analyzing the effect of removals on network $A$. 

\begin{theorem}\label{Uni_Bi_comp}
Consider the set of all operating interdependent networks, namely $G$, with $N_1$ nodes in network $A$, $N_2$ nodes in network $B$, $E$ edges from network $A$ to $B$ and $E$ edges from network $B$ to $A$, where $A$ and $B$ have star topologies and every node has at least one outgoing edge. For any arbitrary value of $D$, the network with largest $\mathcal{MR}(D)$ has bidirectional edges. 
\end{theorem}
\begin{proof}
By contradiction - Let $G_1 \in G$ be the set of bidirectional interdependent networks and $G_2 \in G$ be the set of unidirectional interdependent networks, where $G_1 \cup G_2 = G$. Moreover, for any subset of $D$ nodes in network $B$, namely $Y_D$, let $X(Y_D)$ denote the minimum node removal from network $A$ for the failure of $Y_D$.

Suppose $G_1^* \in G_1$ is the bidirectional network that has the largest $\mathcal{MR}(D)=X^*$ among all networks in $G_1$. Next, we prove by contradiction that there exists no unidirectional interdependent network with larger $\mathcal{MR}(D)$.

Consider an arbitrary unidirectional interdependent network in $G_2$. In order to cause the failure of any subset $Y_D$ with minimum node removal ($X(Y_D)$), one should either remove its direct neighbors $N(Y_D)$ (i.e. the set of nodes in network $A$ that provide direct incoming edges to nodes in $Y_D$) or the nodes that their failure leads to the failure of $N(Y_D)$. Thus, $X(Y_D) \leq N(Y_D)$. Suppose there exists $G_2^* \in G_2$ with $\mathcal{MR}(D)>X^*$. It means that there exists an allocation of $E$ edges from network $A$ to $B$ such that for any $Y_D \in B$, $X^*<X(Y_D) \leq N(Y_D)$. Thus, one can construct a bidirectional network with the same allocation such that $\mathcal{N}(Y_D)>X^*$, for all $Y_D \in B$. Therefore, $\mathcal{MR}(D) = \min \{\mathcal{N}(Y_D): \forall Y_D \in B\} >X^*$ which is a contradiction.

\end{proof}

Theorem \ref{Uni_Bi_comp} indicates that bidirectional networks are more robust than unidirectional networks. Thus, throughout this paper, we will only focus on analyzing networks with bidirectional interdependency.

\begin{lemma}\label{NeighborRemoval}
In bidirectional interdependent networks with star topologies, the smallest set of nodes in network $A$ whose removals lead to the failure of a \textit{given} set of $D$ nodes, namely $Y_D$, in network $B$ is the set of direct neighbors of nodes in $Y_D$, namely $\mathcal{N}(Y_D)$.
\end{lemma}
\begin{proof}
This is due to the fact that in bidirectional interdependent networks with star topologies, failures cascade only in one stage.
\end{proof}
Note that by Lemma \ref{NeighborRemoval}, 
\begin{align}
	\mathcal{MR}(D) & = \min \{\mathcal{N}(Y_D): \forall Y_D \in B, |Y_D|=D \}
\end{align}

By Lemma \ref{NeighborRemoval}, It is easy to see that in bidirectional interdependent networks with star topology, metric $\mathcal{MRB}(D)$ can be obtained directly from metric $\mathcal{MR}(D)$ via equation \ref{Relation_MPFR_OneSided}. Thus, it is enough to only focus on evaluating metric $\mathcal{MR}(D)$

\begin{align}
\mathcal{MRB}(D)& =  \min\{\mathcal{MR}(D-1)+1,\mathcal{MR}(D)\} \nonumber \\
			& =		\min_{i=0,\cdots,D}\{\mathcal{MR}(i)+D-i\} \label{Relation_MPFR_OneSided}
\end{align}

It's worth reminding the readers that all the analysis in the rest of this paper are focused on ``bidirectional'' interdependent networks with ``star'' topology unless mentioned otherwise.

\subsection{Formulation}
Here, we provide an ILP formulation for evaluating metric $\mathcal{MR}(D)$. Let $N_1$ denote the number of nodes in network $A$ and $N_2$ denote the number of nodes in network $B$. Moreover, let $X$ denote the set of binary variables associated to the nodes in network $A$ where $X_i=1$ if node $i$ is removed, and $X_i=0$ otherwise. Similarly, let $Y$ denote the set of binary variables associated to the nodes in network $B$ where $Y_j=1$ if node $j$ fails due to the cascading effect, and $Y_j=0$ otherwise. Our formulation is as follows.

\begin{subequations}
\begin{align}
\min \quad &\sum_{i=1}^{N_1} X_i \label{Objective_ILP}\\
\mbox{s.t.}	\quad & Y_j \leq X_i \quad (i,j) \in E  \label{XYrelation}\\
									\quad & \sum_{j=1}^{N_2} Y_j \geq D  \label{FinalFailureD}\\		
									\quad & X_i, Y_j \in \{0,1\}	\label{IntegerConstraint}
\end{align}
\label{ILP_Star_Bi}
\end{subequations}

Here, the objective is to minimize the number of node removals from network $A$. Constraint (\ref{XYrelation}) shows that node $Y_j$ from network $B$ fails if all of its direct neighbors in network $A$ are removed. Moreover, constraint (\ref{FinalFailureD}) enforces the failure of at least $D$ nodes in network $B$.

\subsection{Complexity}
In this section, we show that evaluating $\mathcal{MR}(D)$ is an NP-complete problem in general; however, for certain values of $D$ it can be solved in polynomial time.

\begin{theorem}\label{MRD_hardness}
For arbitrary values of $D$, finding the $\mathcal{MR}(D)$ in a bidirectional interdependent network with star topology is an NP-complete problem.
\end{theorem}
\begin{proof}
The proof is based on a reduction from the problem of balanced complete bipartite subgraph which is known to be NP-complete \cite{Garey}. The details can be found in Appendix \ref{MRD-hardness-proof}
\end{proof}

\begin{corollary}\label{MRBD_hardness}
For arbitrary values of $D$, finding the $\mathcal{MRB}(D)$ in a bidirectional interdependent network with star topology is an NP-complete problem.
\end{corollary}
\begin{proof}
The proof is based on a reduction from NP-complete problem $\mathcal{MR}(D)$ (Theorem \ref{MRD_hardness}). The details can be found in Appendix \ref{MRBD-hardness-proof}.
\end{proof}

\begin{proposition}\label{Polynomial_MRD}
$\mathcal{MR}(D)$ can be found in polynomial time for values of $D=k$ and $D=N_2-k$ where $k$ is a constant. In particular, for $D=1$, $\mathcal{MR}(D)$ is the minimum degree of nodes in network $B$, and for $D=N_2$, $\mathcal{MR}(D)$ is the size of network $A$; i.e. $N_1$.
\end{proposition}
\begin{proof}
By Lemma \ref{NeighborRemoval}, $\mathcal{MR}(D) = \min \{\mathcal{N}(Y_D): \forall Y_D \in B\}$. For $D=k$ and $D=N_2-k$, the number of combinations of $Y_D$ is polynomial in $D$ (i.e., $C(N_2,k)=C(N_2,N_2-k)=O(N_2^k)$); thus, $\mathcal{MR}(D)$ can be found in polynomial time. 

For $D=1$, clearly the target node in network $B$ is the one with the minimum number of neighbors in network $A$; thus, $\mathcal{MR}(D)$ is the minimum degree of nodes in network $B$.

For $D=N_2$, we prove our claim by contradiction. Suppose node $i \in A$ has not been removed. Since failures cascade only in one stage, removal of no set of nodes in network $A$ can lead to the failure of node $i \in A$. Thus, $i$ remains an operating node, which means that it is connected to at least one node $j \in B$. Therefore, node $j \in B$ is operating, too; i.e. $D < N_2$ which is a contradiction.
\end{proof}

Next, we show that not only one cannot evaluate the exact value of $\mathcal{MR}(D)$ in polynomial time (unless $\mathcal{P=NP}$), one cannot approximate this metric in polynomial time.

\begin{theorem}
There exists no PTAS to provide an r-approximation for the $\mathcal{MR}(D)$ for some values of $r>1$.
\end{theorem}
\begin{proof}
The proof is based on the inapproximability of the balanced biclique problem \cite{feige2004hardness, khot2006ruling}. The details can be found in Appendix \ref{Inapproximability_proof}.
\end{proof}

Next, we show several heuristics that provide nearly-optimal approximations for metric $\mathcal{MR}(D)$ in practice.

\subsection{Heuristics}

In this section, first we propose three heuristics and then, compare their performances using simulation results.

\subsubsection{Greedy Algorithm}
The first algorithm is a Greedy approach that only uses the adjacency matrix of the network, and works as follows.

\begin{table}[h]
	\begin{tabular}{ll}	
	\hline
			& \textbf{Greedy Algorithm}\\
	\hline
			1 & Initialize the removal and failure sets as $R=\phi$ and $F=\phi$;\\
			2 & Select the node with minimum degree in network $B$, and add it to $F$;\\
				& If there are several nodes with minimum degree, select one randomly;\\
			3 & Remove all nodes in network $A$ that are attached to the node selected in \\
			  & Step 2. Add these nodes to set $R$;\\
			4 & Remove all the edges attached to the nodes in $F$ and $R$. Update degrees;\\
			5 & Repeat previous steps until $|F|=D$;\\
			6 & Return $|R|$.	\\
	\hline	
	\end{tabular}
\end{table}

In each iteration, the greedy algorithm removes the minimum number of nodes from network $A$ required for the the failure of one additional node in network $B$. Therefore, after at most $D$ iterations, removal of nodes in set $R$ leads to the failure of $D$ nodes in network $B$; i.e. set $F$.

\begin{proposition}\label{Greedy_Performance}
In the worst case, the solution of greedy algorithm is no more than $D$ times the optimal solution.
\end{proposition}
\begin{proof}
By contradiction - Suppose that $\mathcal{MR}(D)=X$; thus, the degree of each node in the optimal failure set in network $B$ is at most $X$. Moreover, suppose that the greedy algorithm returns a removal set of size $X'$ where $X' > DX$. Thus, greedy algorithm has selected a node in network $B$ with degree of larger than $X$. This is contradiction to the fact that greedy starts by selecting nodes in network $B$ with minimum degree, and there are at least $D$ nodes with degree smaller than or equal to $D$.

Next, we show that this bound can be tight. Consider a bipartite graph where network $A$ has $X(D+1)$ nodes divided into $D+1$ batches of equal sizes, and network $B$ has $2D$ nodes divided into two batches of equal sizes. Connect each node $i$ in the first batch of network $B$ to all the $X$ nodes in the $i^{th}$ batch in network $A$. Moreover, connect all of the $D$ nodes in the second batch in network $B$ to all of the $X$ nodes in the last batch in network $A$ (See Figure \ref{WorstCase}). It is easy to see that $\mathcal{MR}(D)=X$ where the greedy algorithm could select $XD$ nodes. This is due to the fact that all nodes in network $B$ have degree $X$. Thus, greedy algorithm could select all nodes from the first batch in network $B$ which requires $XD$ removals from network $A$.

\begin{figure}[ht]
\centering
\includegraphics[scale=0.4]{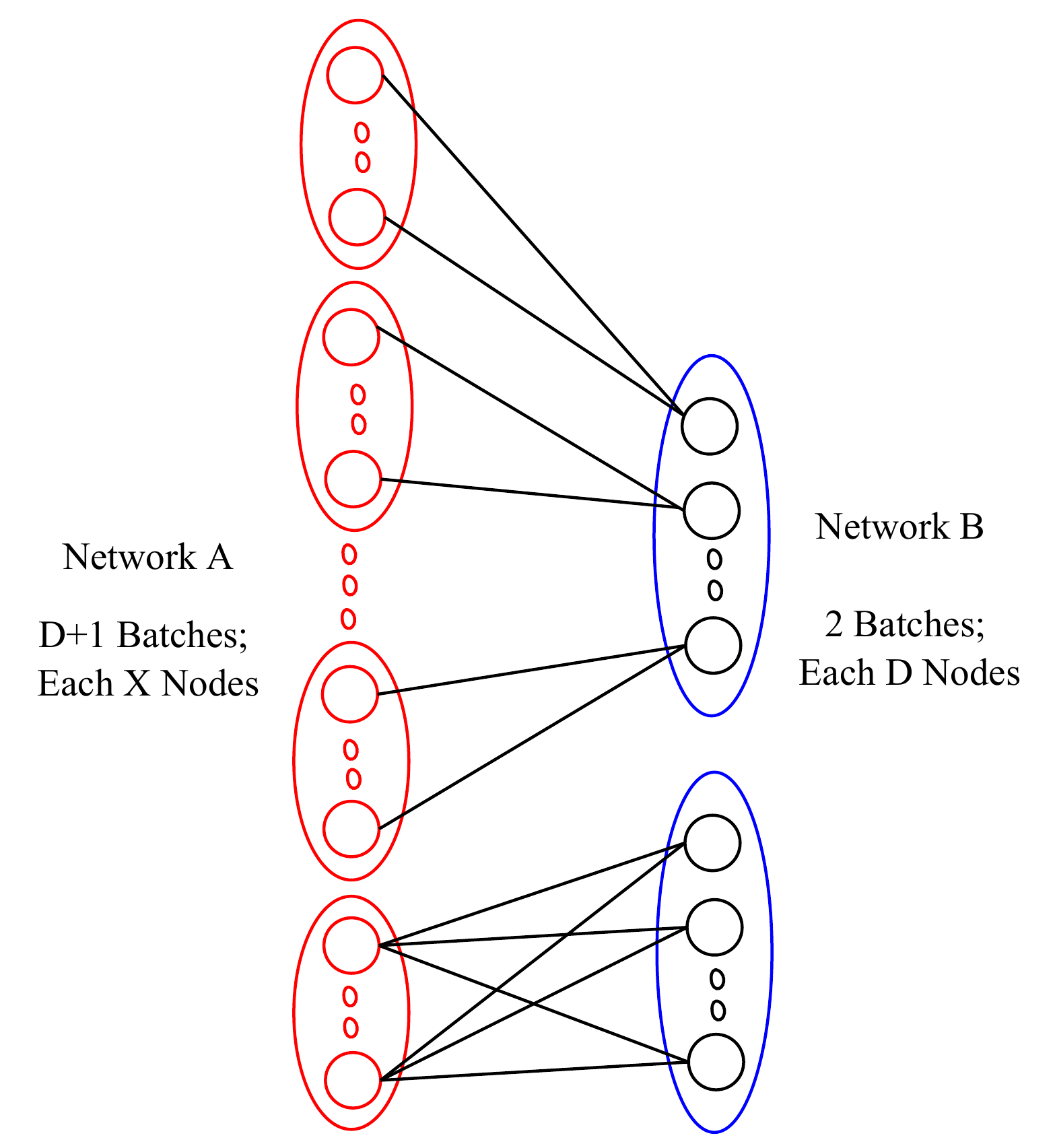}
\caption{A scenario where worst-case bound of greedy algorithm is tight.}
\label{WorstCase}
\end{figure}

\end{proof}

Note that in the example of Proposition \ref{Greedy_Performance}, all nodes in network $B$ have equal degrees, and the greedy algorithm selects one of them randomly. Therefore, although this algorithm could achieve the worst-case solution, the probability of this event is $(1-\frac{D}{2D})(1-\frac{D-1}{2D-1})\cdots(1-\frac{1}{D+1})$ which becomes very small as $D$ increases. Later, in the simulation section, we will show that the greedy algorithm has a good performance in most scenarios.

\subsubsection{Randomized Rounding}

The second algorithm is a modified randomized rounding. Randomized rounding is a widely used technique to solve difficult integer optimization problems. In general, it solves the Linear Program (LP) relaxation of the original ILP formulation, and rounds the solution randomly. In our case, we relax the constraint (\ref{IntegerConstraint}) so that $X$ and $Y$ can take any real value in range $[0,1]$. 

Let $X_i^*$ and $Y_j^*$ be the optimal values of the relaxed LP problem. Our randomized rounding algorithm works as follows.

\begin{table}[h]
	\begin{tabular}{ll}	
	\hline
			& \textbf{Randomized Rounding}\\
	\hline
			1 & Initialize the removal and failure sets as $R=\phi$ and $F=\phi$; \\
			2 & Select each node $Y_j \in B$ with probability $Y_j^*$, and add it to set $F$;\\
			3 & Repeat step 2 until $|F|=D$; i.e. $D$ nodes fail;\\
			4 & Find all the nodes in network $A$ that are attached to the nodes in the failure\\
			  & set $F$. Add all these nodes the set of removals $R$;\\
			6 & Return $|R|$.\\
	\hline	
	\end{tabular}
\end{table}

In this algorithm, we select nodes from network $B$ randomly and independently until $D$ nodes are selected for the failure set $F$. Clearly, nodes with larger values of $Y_j^*$ have a higher chance to be part of set $F$. Later, we will see in the simulation section that for networks that $Y_j^*$ has values close to either 1 or 0, the randomize rounding algorithm provides a nearly-optimal solution.

\subsubsection{Simulated Annealing}
Simulated Annealing is a random search strategy that can be used to find the near optimal solutions for integer problems \cite{Wolsey1998}. Here, we propose two versions of the SA where the difference is in selecting the neighbors.

The first algorithm selects a random neighbor $R'$ of current removal set $R$ by adding, removing or replacing nodes in $R$, and then checking for feasibility; i.e. checking if the new removal set $R'$ leads to the failure of $D$ nodes in network $B$. If the neighbor set $R'$ is feasible and has smaller or equal number of nodes, algorithm moves to $R'$ with probability 1. If $R'$ is feasible but larger; i.e. has an additional node $i \in R' \backslash R$, the algorithm moves to $R'$ with some positive probability proportional to the degree of node $i$ such that neighbors with larger degree nodes are more likely to be selected.

Let $d(i)$ denote the degree of node $i$. The details of first SA algorithm are as follows.

\begin{table}[h]
	\begin{tabular}{ll}	
	\hline
			& \textbf{Simulated Annealing 1}\\
	\hline
			1 & Start with an initial set of node removals $R=R_0$ from $A$ that lead to the \\
			  & failure of $D$ nodes in $B$; Set initial temperature $T$, final temperature $T_F$,  \\
			  & and reduction parameter $r \in (0,1)$; \\
			2 & Repeat the followings for $L$ times:\\
			  & a) Pick a neighbor of $R$, namely $R'$, by either adding, removing or  \\
			  &    replacing one random node in $R$; \\			  
			  & b) set $\Delta=1$, if $|R'|>|R|$; and set $\Delta=-1$, otherwise;\\
			  & c) If $R'$ is feasible; i.e. removal of nodes in $R'$ leads to the failure of $D$ \\
			  &    nodes in $B$, move to the new neighbor according to the following rules:\\
			  & \hspace{2mm} i) if $\Delta = -1$, set $R=R'$ and $F=F'$;\\
			  & \hspace{2mm} ii) if $\Delta = 1$, set $R=R'$ and $F=F'$ with probability \\
			  & \hspace{2mm}     $exp(-\frac{1}{T} (1-\frac{d(i)}{\sum_{i=1}^{N_1} d(i)}))$;\\
			3 & Set $T=rT$;\\
			4 & Repeat steps 2 and 3 until $T < T_F$;\\
			5 & Return $|R|$.\\
	\hline	
	\end{tabular}
\end{table}

Next, we propose another Simulated annealing algorithm which selects a random neighbor $F'$ of failure set $F$ such that $|F'| \geq D$. This guarantees that the removal set $R'$ associated to failure set $F'$ is always feasible. Under this selection, if $R'$ has smaller or equal number of nodes than $R$, the algorithm moves to the new neighbor; otherwise, it moves to the larger neighbor with some positive probability proportional to the increase in size of removal set, where larger $R'$ has lower probability to be selected. The details of algorithm is as follows.

\begin{table}[h]
	\begin{tabular}{ll}	
	\hline
			& \textbf{Simulated Annealing 2}\\
	\hline
			1 & Start with an initial set of node removals $R=R_0$ from $A$ that lead to the \\
			  & failure of $D$ nodes in $B$; Set initial temperature $T$, final temperature $T_F$, \\
			  & and reduction parameter $r \in (0,1)$; \\
			2 & Repeat the followings for $L$ times:\\
			  & a) Pick a \textit{feasible} neighbor of $F$, namely $F'$, according to the following rules:\\
			  & \hspace{2mm} i) if $|F|=D$, either add or replace a random node in $F$ (call it $F'$), \\
			  & \hspace{2mm}    and find the set of removals $R'$ for failure of $F'$;\\			  
			  & \hspace{2mm} ii) if $|F|>D$, randomly add or remove a node from $F$ (call it $F'$),\\
			  & \hspace{2mm}    and find the set of removals $R'$ for failure of $F'$;\\
			  & b) set $\Delta=1$, if $|R'|>|R|$; and set $\Delta=-1$, otherwise;\\
			  & c) Move to the new neighbor according to the following rules:\\
			  & \hspace{2mm} i) if $\Delta = -1$, set $R=R'$ and $F=F'$;\\
			  & \hspace{2mm} ii) if $\Delta = 1$, set $R=R'$ and $F=F'$ with probability \\
			  & \hspace{2mm}     $exp(-\frac{|R'|-|R_0|}{T})$;\\
			3 & Set $T=rT$;\\
			4 & Repeat steps 2 and 3 until $T < T_F$;\\
			5 & Return $|R|$.\\
	\hline	
	\end{tabular}
\end{table}

In practice, we initialize both simulated annealing algorithms with the solution of greedy algorithm to have a good starting point. In addition, instead of returning the final removal set, the algorithm returns the smallest $|R|$ found during all iterations.

\subsubsection{Comparison}

In this section, we compare the performances of our algorithms by running simulations over a set of networks. We also obtain the optimal solution by solving the ILP formulation given by equations (\ref{Objective_ILP})-(\ref{IntegerConstraint}) using CPLEX. The ILP can be solved for small networks; thus, we can compare the performance of our algorithms with the optimal solution.

Since the networks in this paper have bipartite topologies, we generate random bipartite graphs according to the Molloy and Reed model, where every pair of nodes are randomly connected based on the degree of all nodes (See \cite{guillaume2006bipartite} for more details). Here, we consider networks with two types of degree distributions: Type (1) all $N$ nodes on both sides have a binomial degree distribution with average $k$\footnote{We also generated random regular bipartite graphs with degree $k$; since the behavior of regular graphs was very close to graphs with binomial distribution, we do not show the simulation results}, and Type (2) half of nodes in each side has a binomial degree distribution with average $k_1$ and the other half has a binomial distribution with average $k_2$.

Figures \ref{MRD_N100_Binomial_D1_D5_1}-\ref{MRD_N100_Binomial_D1_D5_4} show the performances of our algorithms for type(1) networks of size $N=100$ and average degrees $k=1,\cdots,4$. It can be seen that for $k=1$, the randomized rounding is nearly optimal. However, as $k$ increases, its performance degrades. This is due to the fact that for small degree networks, the optimal solution of the relaxed LP achieves values close to 0 or 1. Thus, approximating these values will give a nearly optimal solution. However, as the degree increases, the values of variables in the relaxed LP are no longer close to 0 or 1; thus, the approximation of these values is no longer close to the optimal solution.

Moreover, as $k$ increases, the performance of the greedy algorithm improves. The reason is that in networks with small degrees, there are fewer nodes in network $B$ that have common neighbors in $A$. Therefore, the greedy algorithm has a lower chance to find them. However, when the degree increases, more nodes share neighbors; thus, the greedy algorithm performs better (See Proposition \ref{Greedy_Performance} for a more detailed argument).

Finally, as expected both simulated annealing algorithms perform better than greedy algorithm. This is due to the fact that the starting point of the simulated annealing algorithm is selected to be the output of the greedy algorithm.

\begin{figure}[ht]
\centering
\subfigure[k=1]
{\label{MRD_N100_Binomial_D1_D5_1}\includegraphics[scale=0.035]{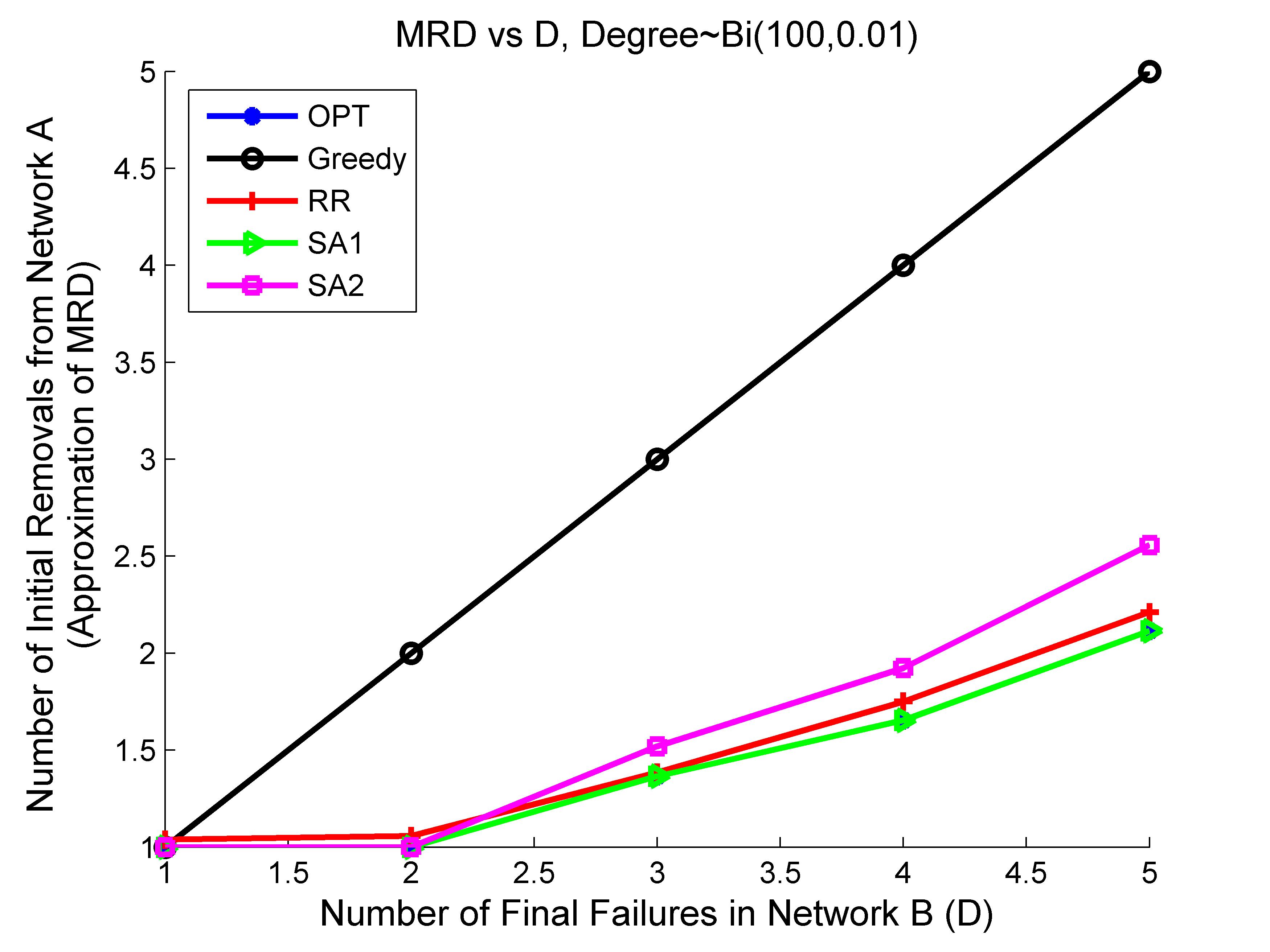}\vspace{-0.3cm}}
\subfigure[k=2]
{\label{MRD_N100_Binomial_D1_D5_2}\includegraphics[scale=0.035]{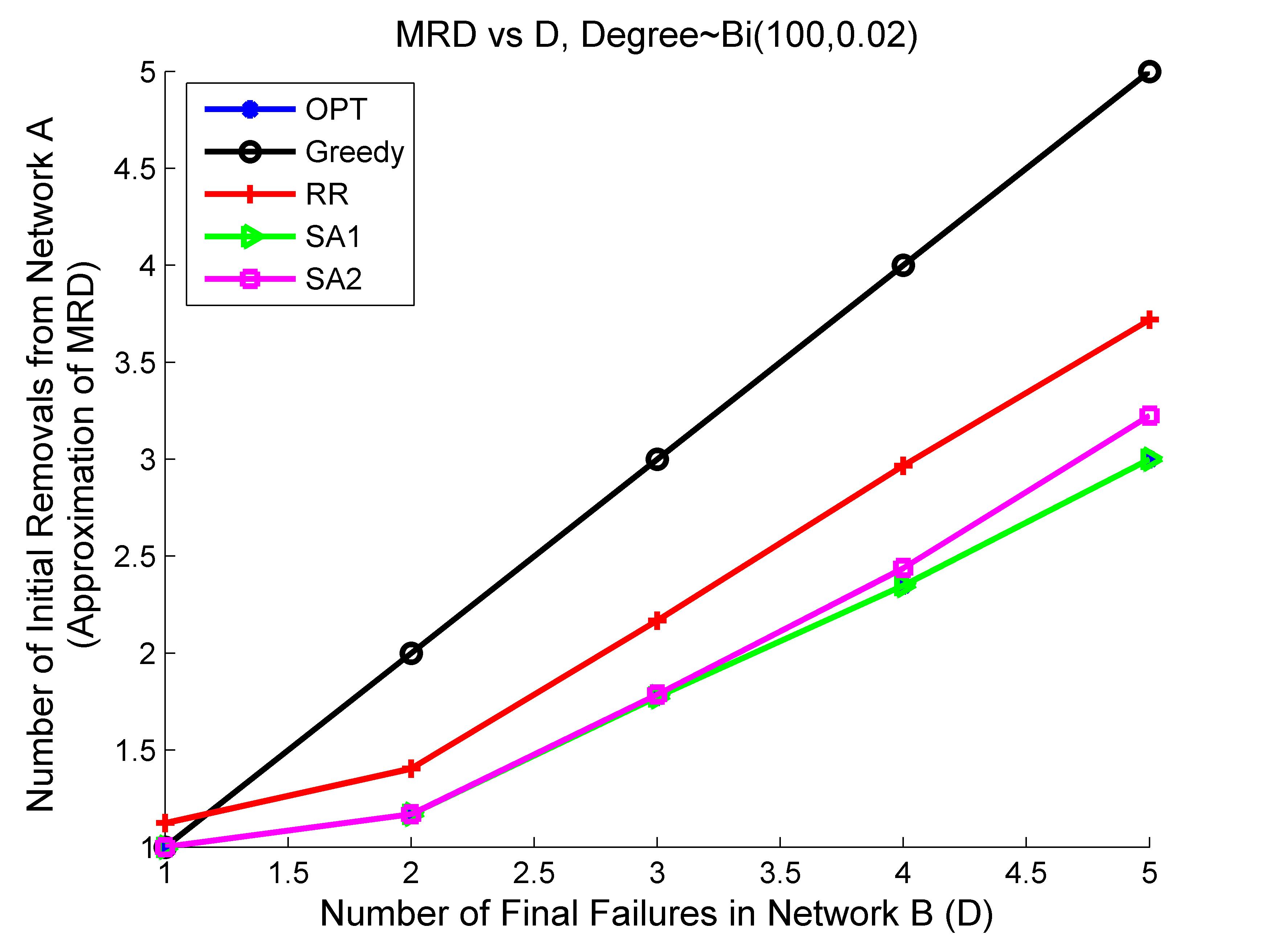}\vspace{-0.3cm}}
\subfigure[k=3]
{\label{MRD_N100_Binomial_D1_D5_3}\includegraphics[scale=0.035]{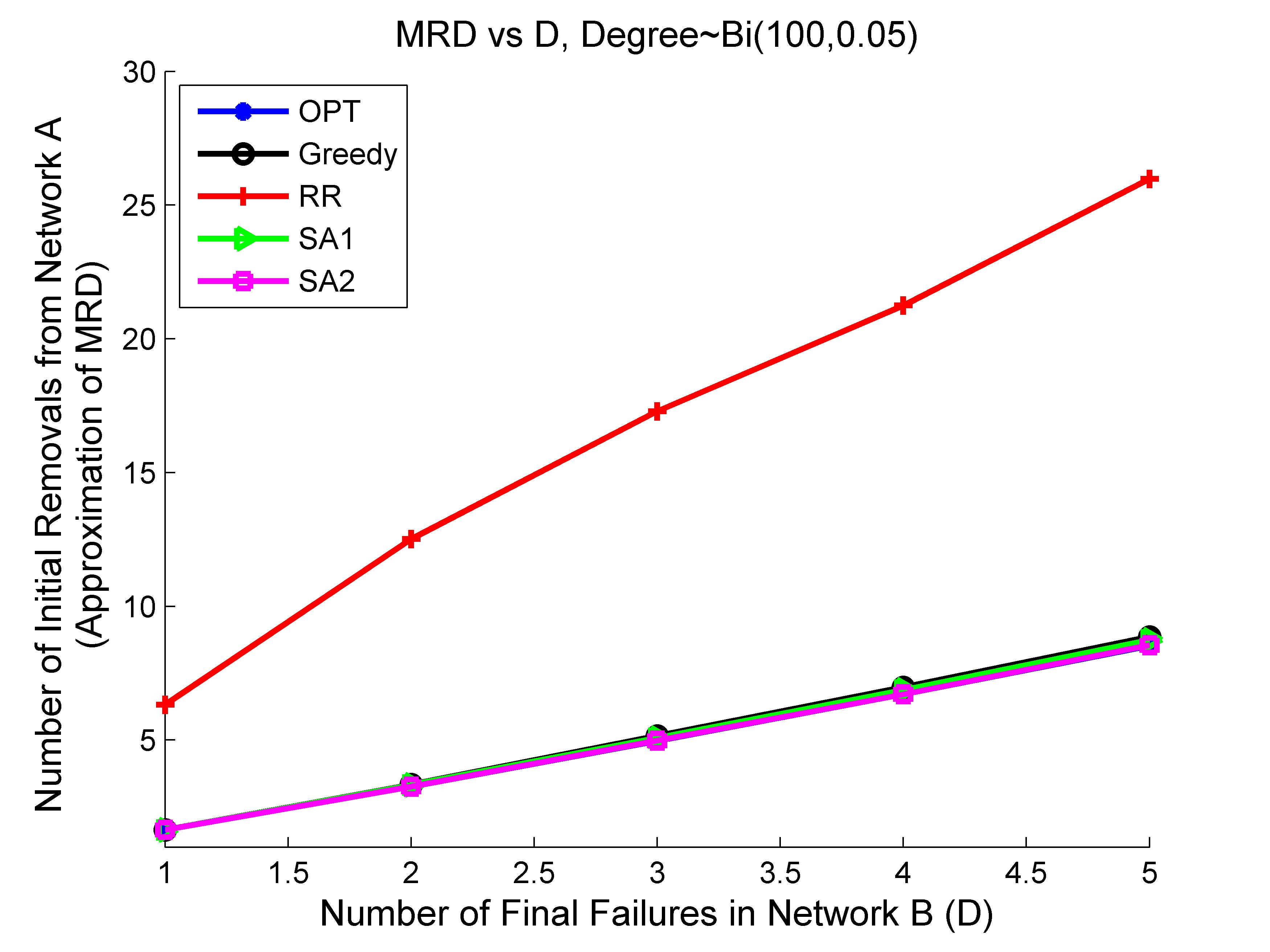}\vspace{-0.3cm}}
\subfigure[k=4]
{\label{MRD_N100_Binomial_D1_D5_4}\includegraphics[scale=0.035]{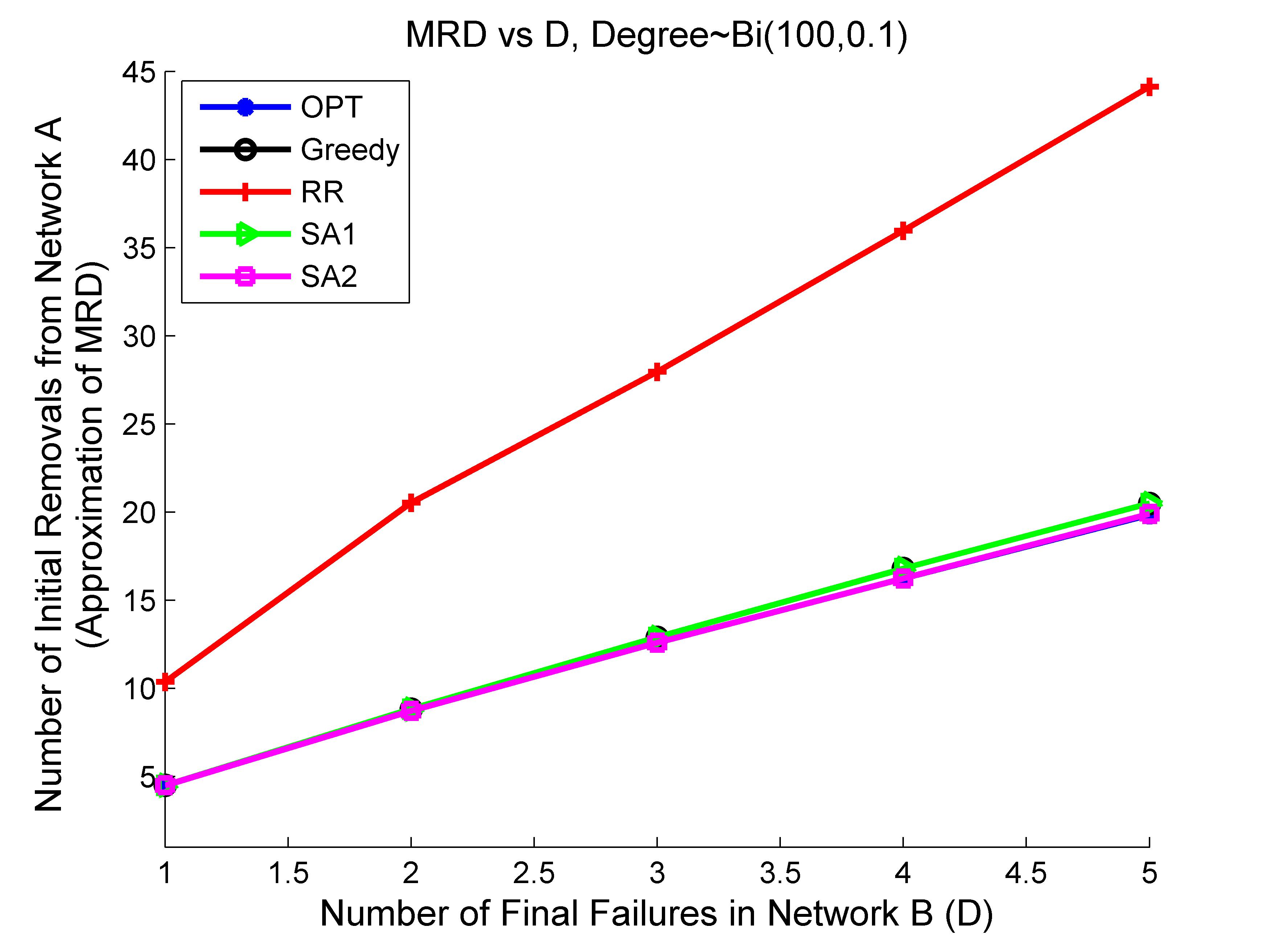}\vspace{-0.3cm}}
\caption{Minimum Node Removal vs Final Failure Size, Type(1) network of size $N=100$, Failure sizes $D \in [1,2,3,4,5]$} \label{MRD_N100_Binomial_D1_D5}\vspace{-0.2cm}
\end{figure}

Figures \ref{Time_N100_Binomial_D1_D5_1}-\ref{Time_N100_Binomial_D1_D5_4} show the runtime of the algorithms for the same set of networks. It can be seen that greedy and randomized rounding are very fast, and the runtime for the optimal solution becomes prohibitive as the size of the network increases. Moreover, it can be seen that the first simulated annealing algorithm has an almost constant run time for all values of average degree $k$ and final failures $D$, whereas the runtime of the second simulated annealing algorithm remains constant for all values of average degree $k$, but increases as $D$ increases.

\begin{figure}[ht]
\centering
\subfigure[k=1]
{\label{Time_N100_Binomial_D1_D5_1}\includegraphics[scale=0.035]{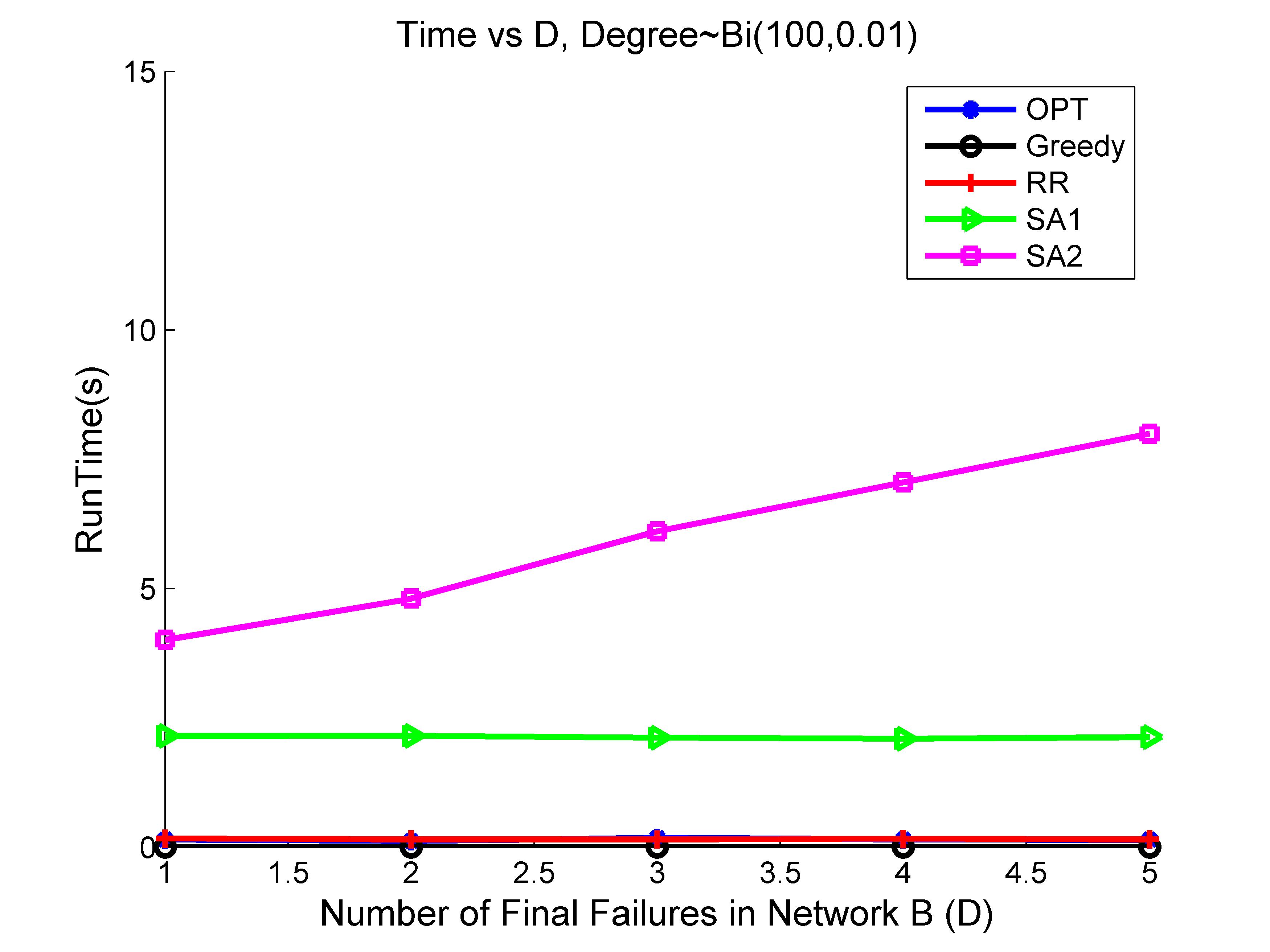}\vspace{-0.3cm}}
\subfigure[k=2]
{\label{Time_N100_Binomial_D1_D5_2}\includegraphics[scale=0.035]{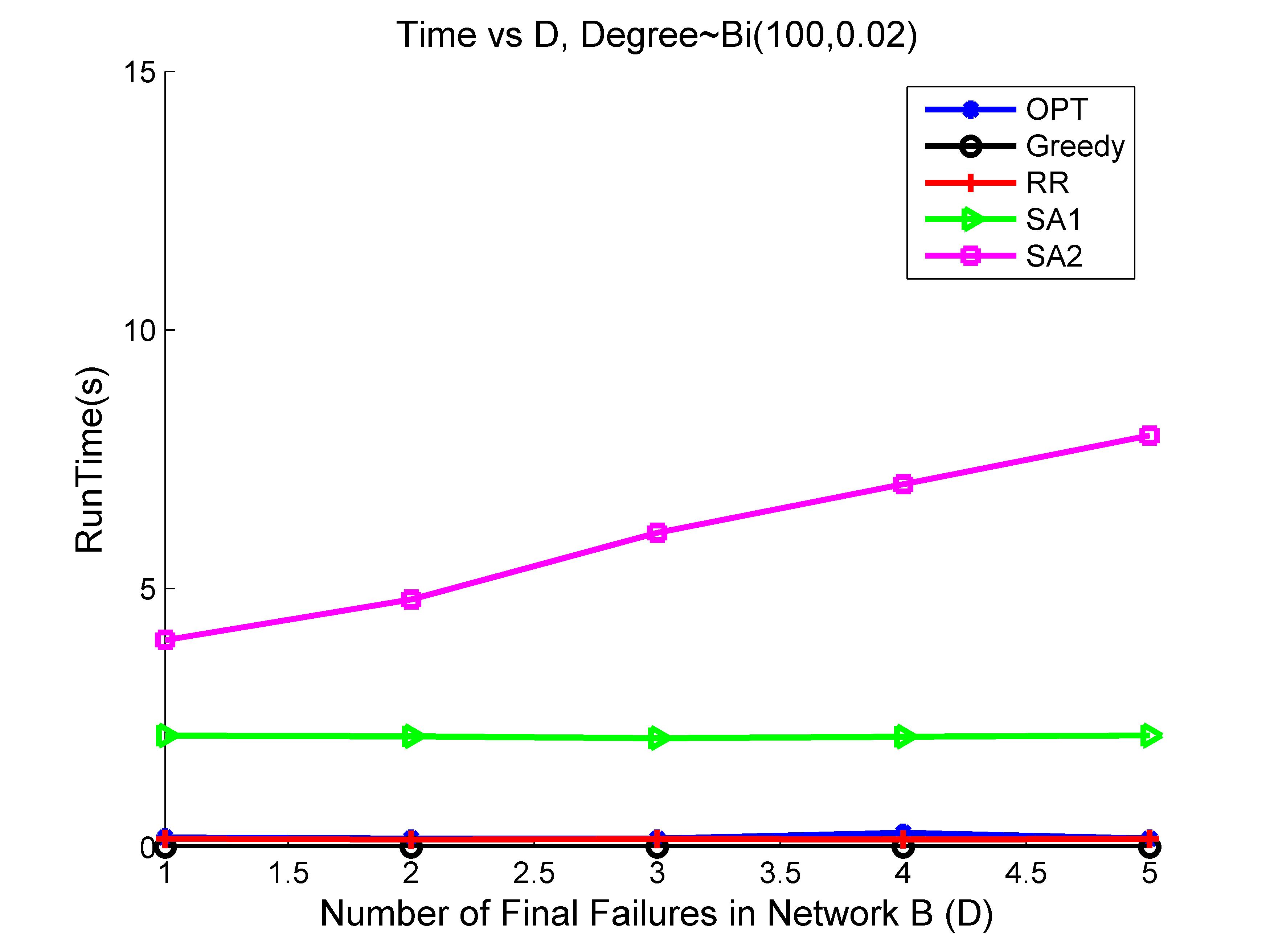}\vspace{-0.3cm}}
\subfigure[k=3]
{\label{Time_N100_Binomial_D1_D5_3}\includegraphics[scale=0.035]{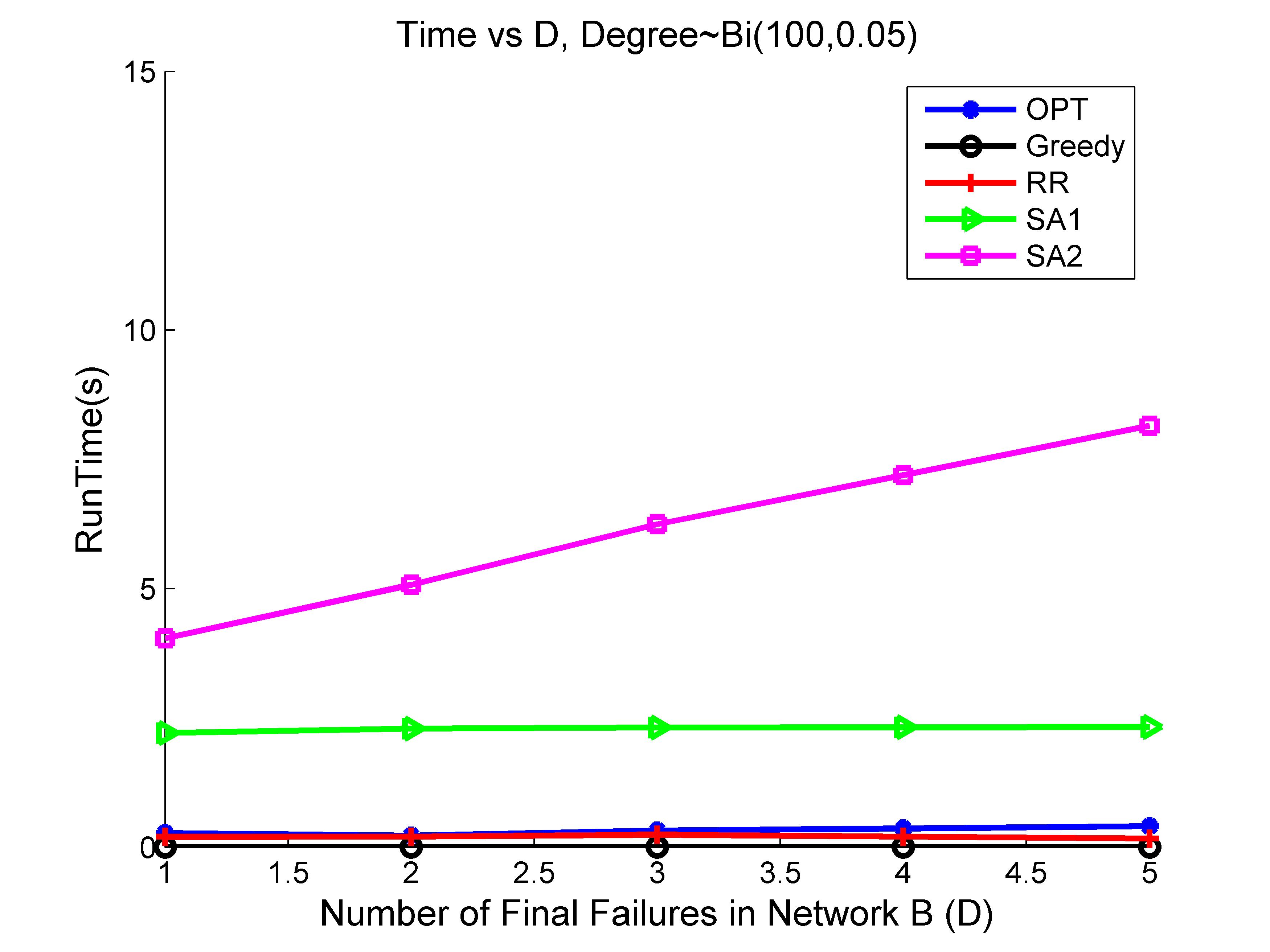}\vspace{-0.3cm}}
\subfigure[k=4]
{\label{Time_N100_Binomial_D1_D5_4}\includegraphics[scale=0.035]{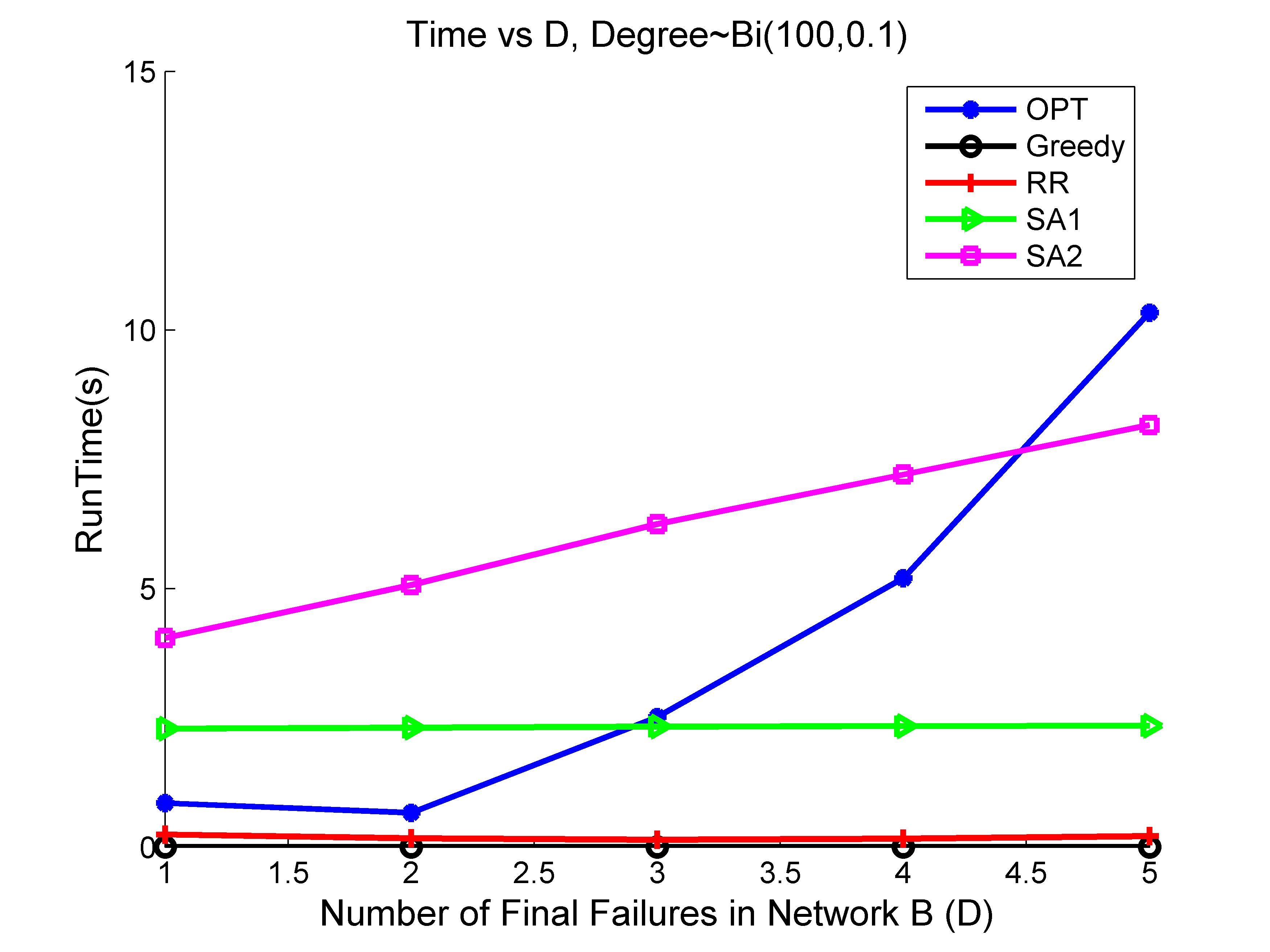}\vspace{-0.3cm}}
\caption{Run-time vs Final Failure Size, Type(1) network of size $N=100$, Failure sizes $D \in [1,2,3,4,5]$}
\label{Time_N100_Binomial_D1_D5}\vspace{-0.2cm}
\end{figure}

We also analyze the performances of our algorithms for the same set of networks but larger values of $D$. Figures \ref{MRD_N100_Binomial_D45_D50_1}-\ref{MRD_N100_Binomial_D45_D50_4} show that the behavior of the algorithms remains the same, and the the first simulated annealing algorithm performs nearly-optimal.

\begin{figure}[ht]
\centering
\subfigure[k=1]
{\label{MRD_N100_Binomial_D45_D50_1}\includegraphics[scale=0.035]{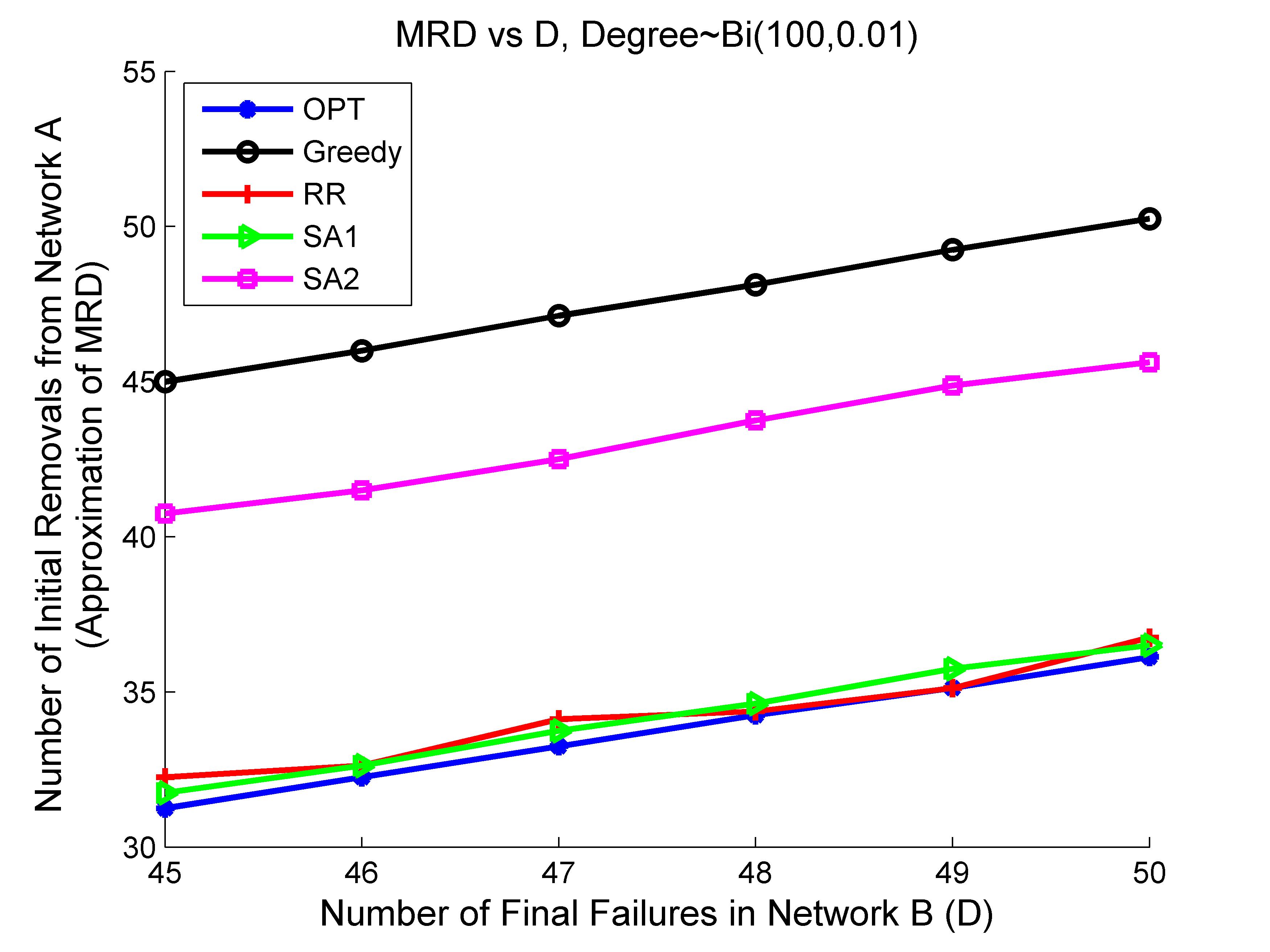}\vspace{-0.3cm}}
\subfigure[k=2]
{\label{MRD_N100_Binomial_D45_D50_2}\includegraphics[scale=0.035]{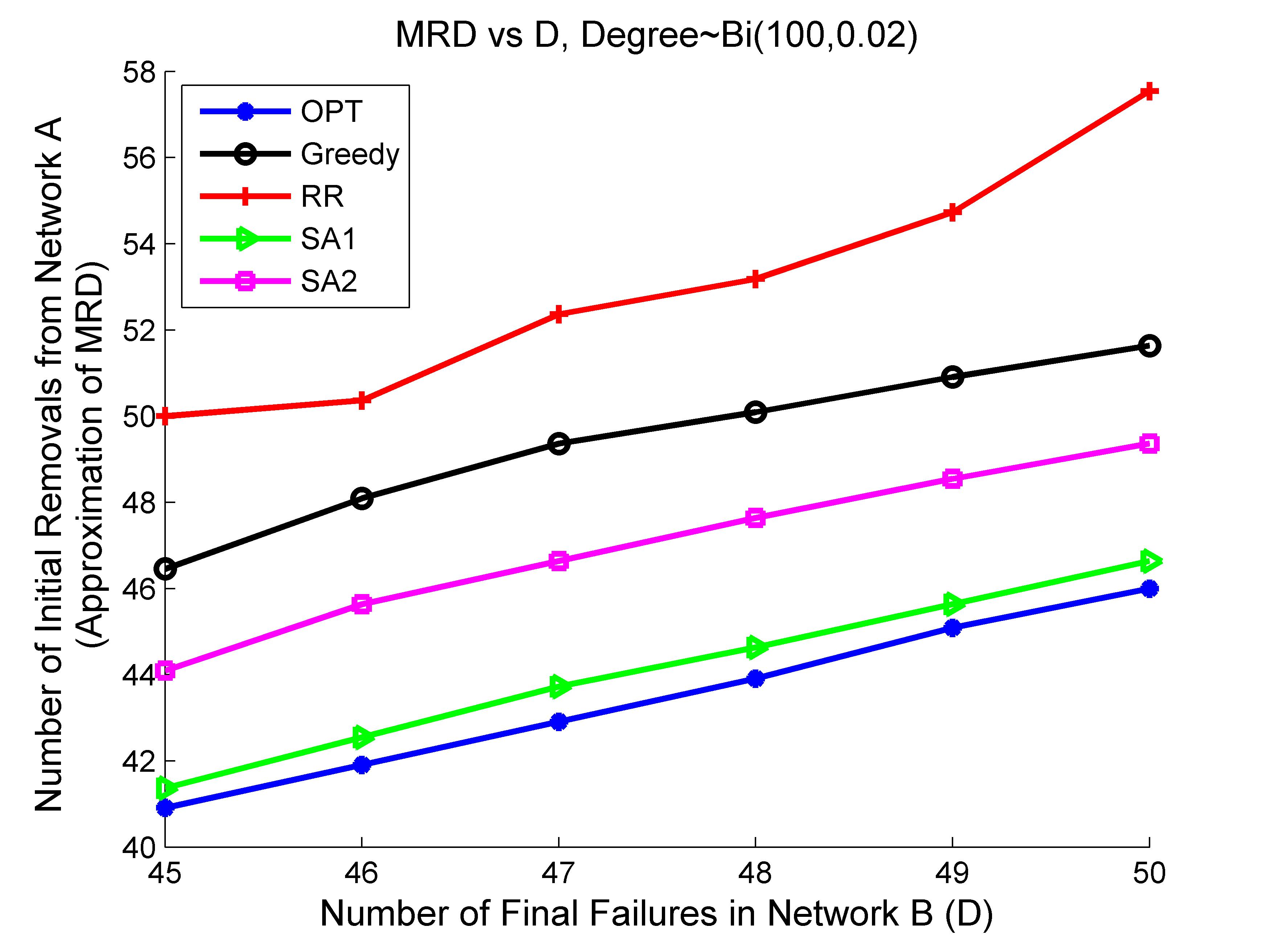}\vspace{-0.3cm}}
\subfigure[k=3]
{\label{MRD_N100_Binomial_D45_D50_3}\includegraphics[scale=0.035]{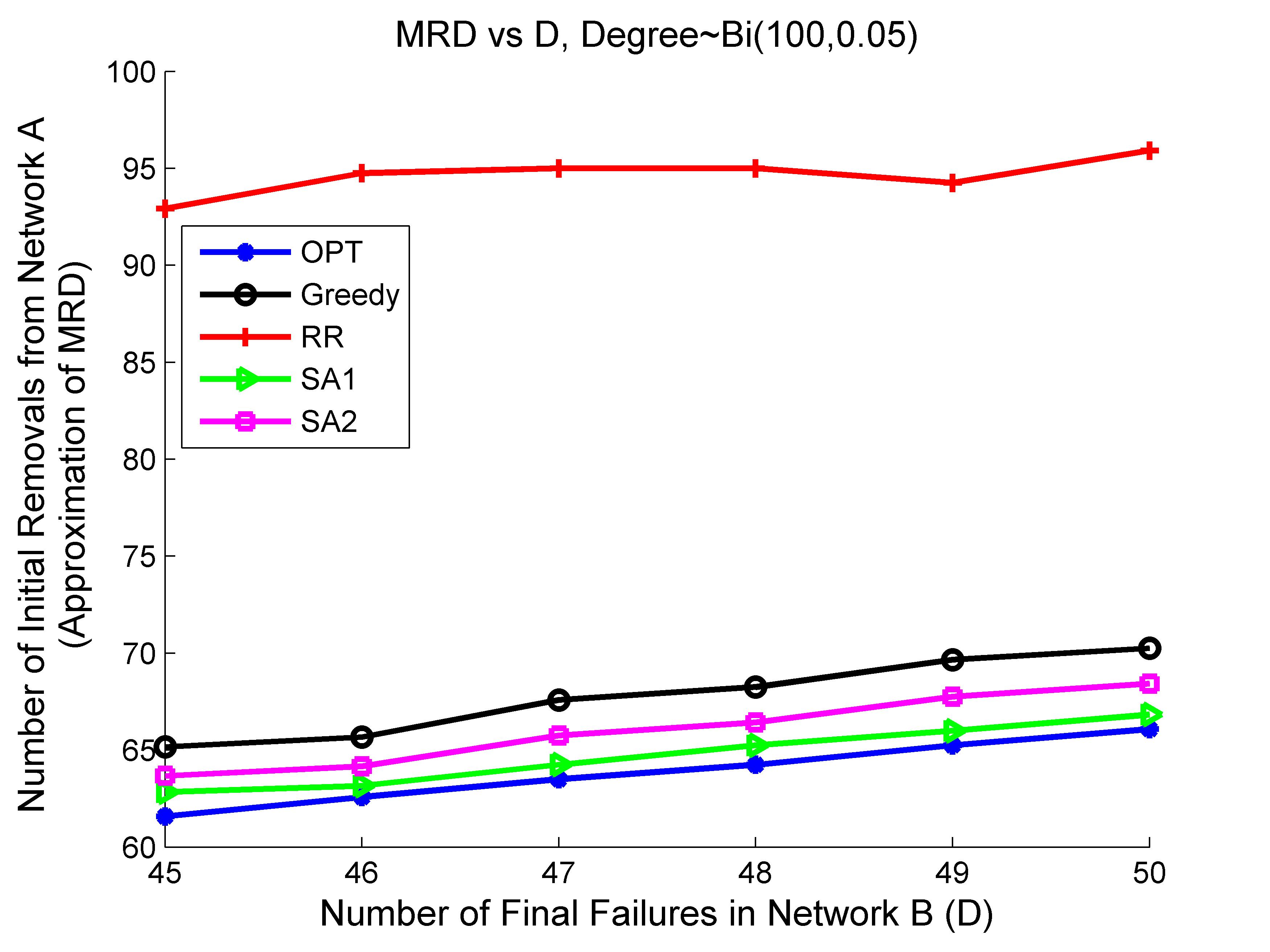}\vspace{-0.3cm}}
\subfigure[k=4]
{\label{MRD_N100_Binomial_D45_D50_4}\includegraphics[scale=0.035]{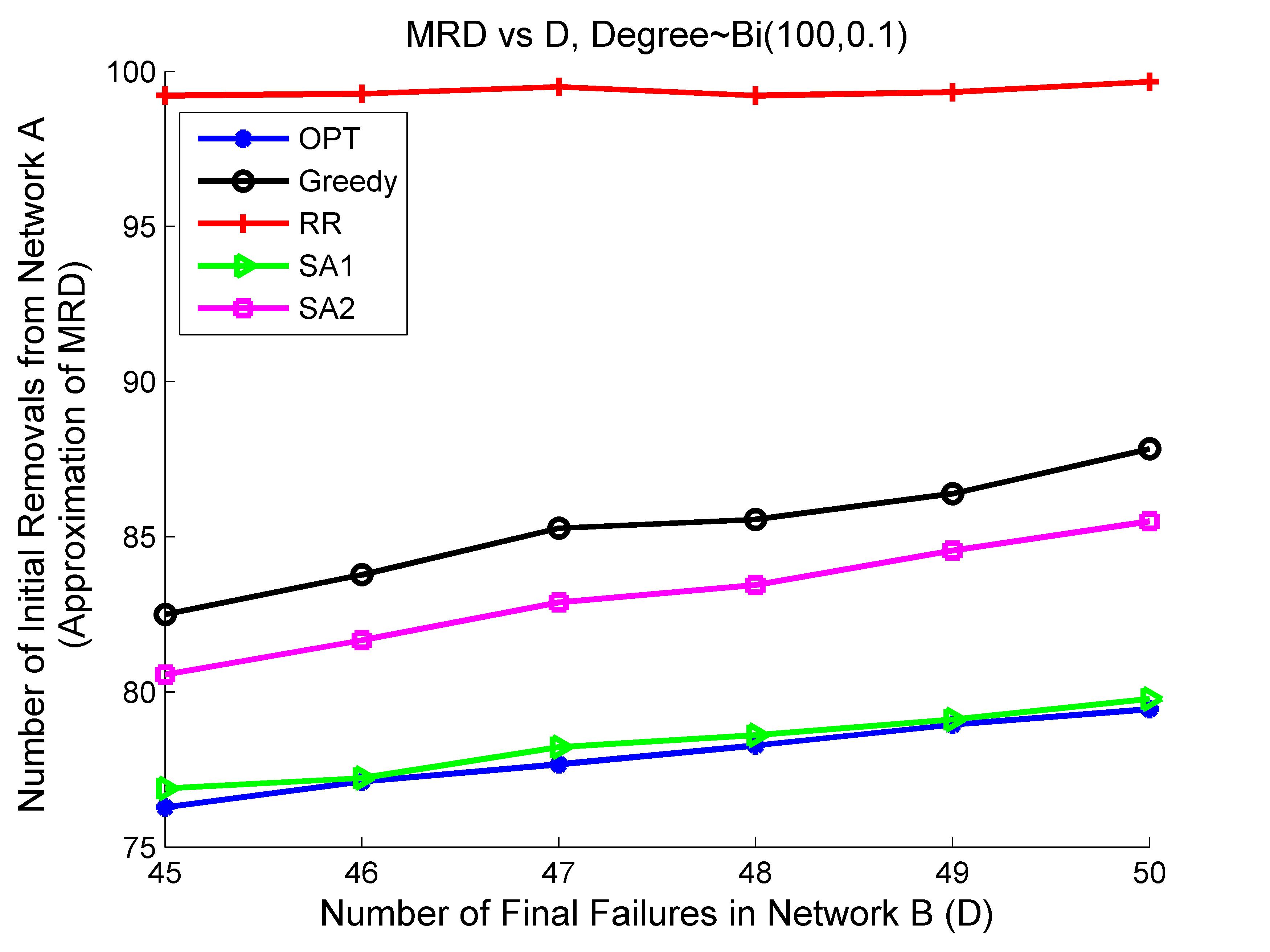}\vspace{-0.3cm}}
\caption{Minimum Node Removal vs Final Failure Size, Type(1) network of size $N=100$, Failure sizes $D \in [45,46,47,48,49,50]$}
\label{MRD_N100_Binomial_D45_D50}\vspace{-0.2cm}
\end{figure}

Next, we analyze the performances of our algorithms for type(2) networks of size $N=100$ and average degrees of $k_1=2$ and $k_2=20$. It can be seen from Figure \ref{MRD_N100_TwoBinomial} that the first simulated annealing algorithm provides the best performance. We also observed that the randomized rounding algorithm performs poorly in this set of networks.

\begin{figure}[ht]
	\centering
	\includegraphics[scale=0.05]{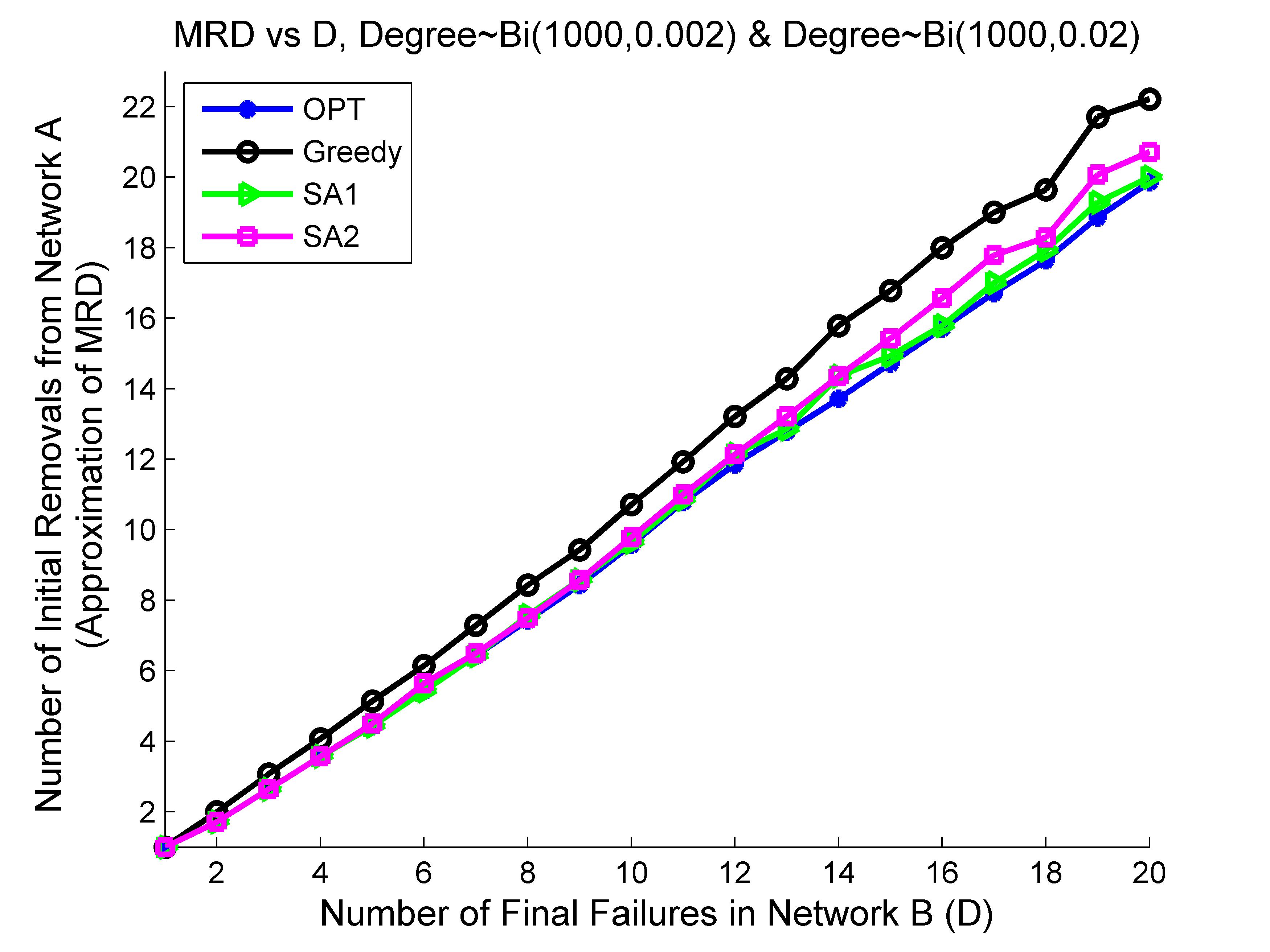}
	\caption{Minimum Node Removal vs Final Failure Size, Type(2) network of size $N=100$ and $k_1,k_2=[2,20]$, Failure sizes $D \in [1,\cdots,20]$}
	\label{MRD_N100_TwoBinomial}
	\vspace{-0.2cm}
\end{figure}

Finally, we consider larger networks of size $N=1000$. For this size of network, the ILP formulation cannot be solved optimally anymore as the run-time becomes prohibitive. Thus, we only compare the performances of the heuristic algorithms. Figures \ref{LargeNetwork_Binomial} and \ref{LargeNetwork_TwoBinomial} illustrate the results of networks of type(1) and type(2). It can be seen that the simulated annealing algorithms do not provide a significant improvement in the size of initial removals compared to the greedy algorithm, while their run time is much larger than the greedy algorithm.

Another interesting point is that for large networks, the second simulated annealing algorithm outperforms the first one. Moreover, the run time of the second simulated annealing algorithm remains constant, while the run time of the first simulated annealing algorithm increases as $D$ increases.

\begin{figure}[ht]
\centering
\subfigure[Minimum Node Removal vs Final Failure Size]
{\label{MRD_LargeNetwork_Binomial}\includegraphics[scale=0.25]{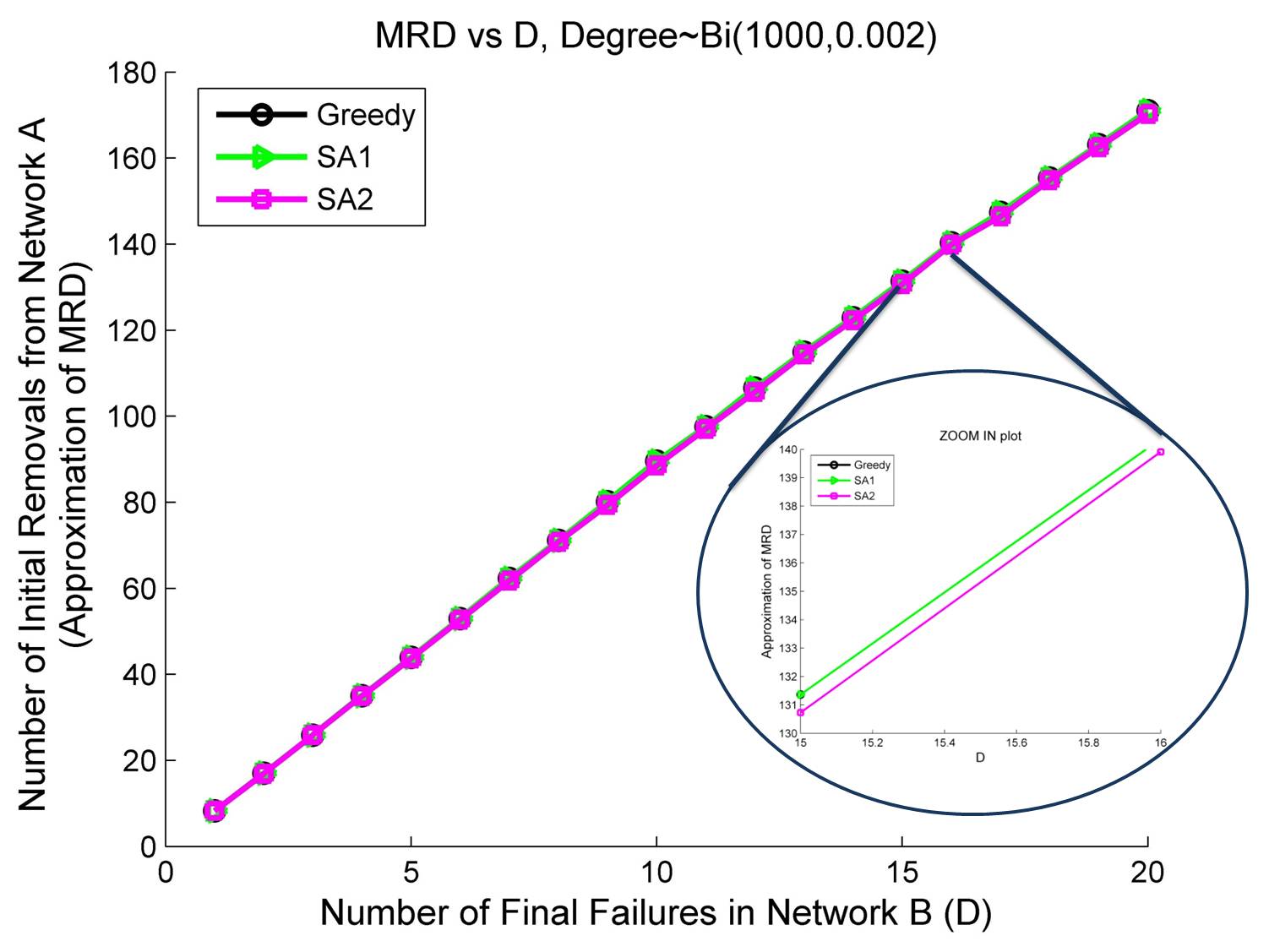}\vspace{-0.3cm}}
\subfigure[Run-Time vs Final Failure Size]
{\label{Time_LargeNetwork_Binomial}\includegraphics[scale=0.05]{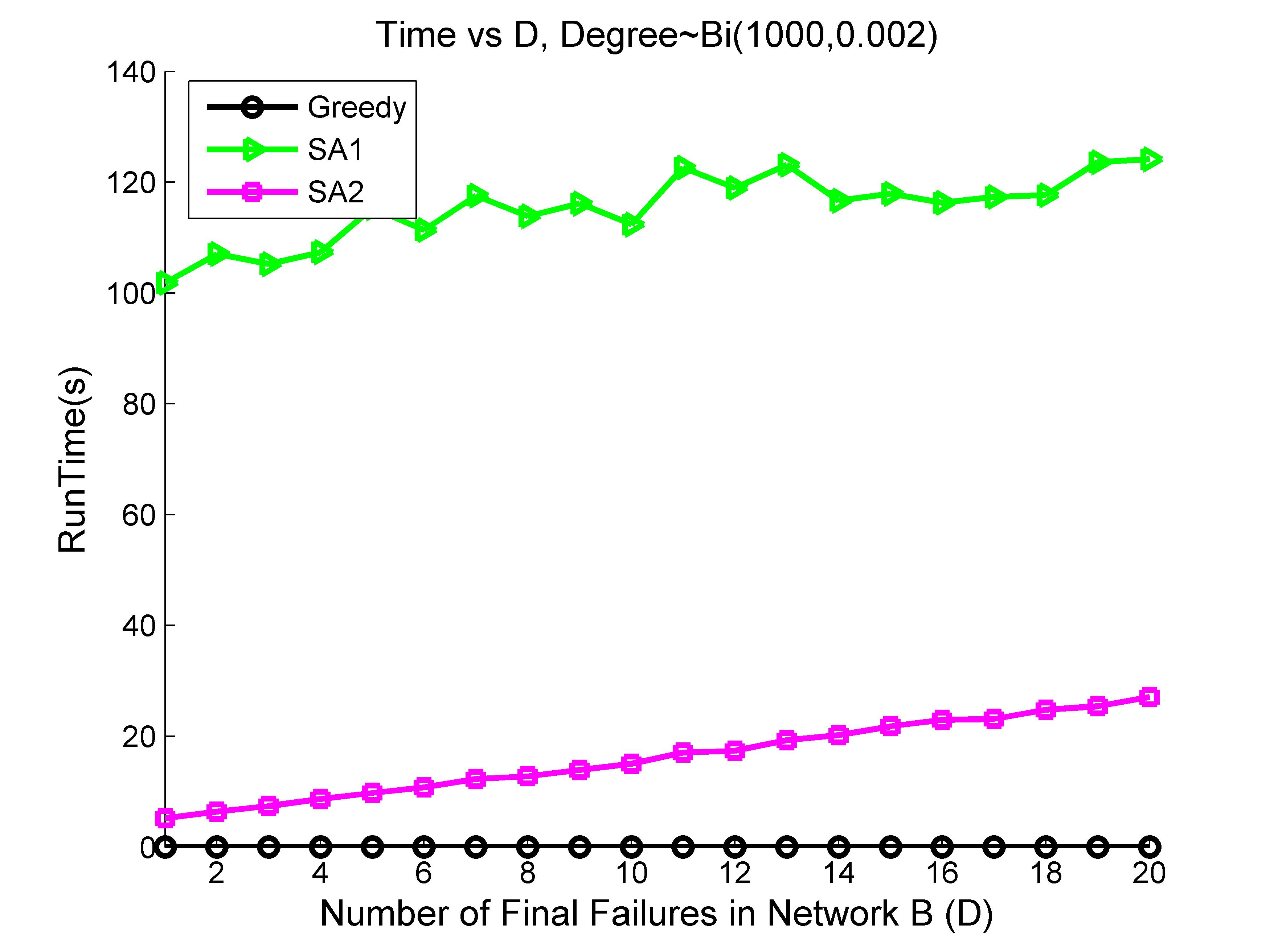}\vspace{-0.3cm}}
\caption{Minimum Node Removal and Run-Time vs Final Failure Size, Type(1) network of size $N=1000$, Failure sizes $D \in [1,\cdots,20]$}
\label{LargeNetwork_Binomial}\vspace{-0.2cm}
\end{figure}

\begin{figure}[ht]
\centering
\subfigure[Minimum Node Removal vs Final Failure Size]
{\label{MRD_LargeNetwork_TwoBinomial}\includegraphics[scale=0.25]{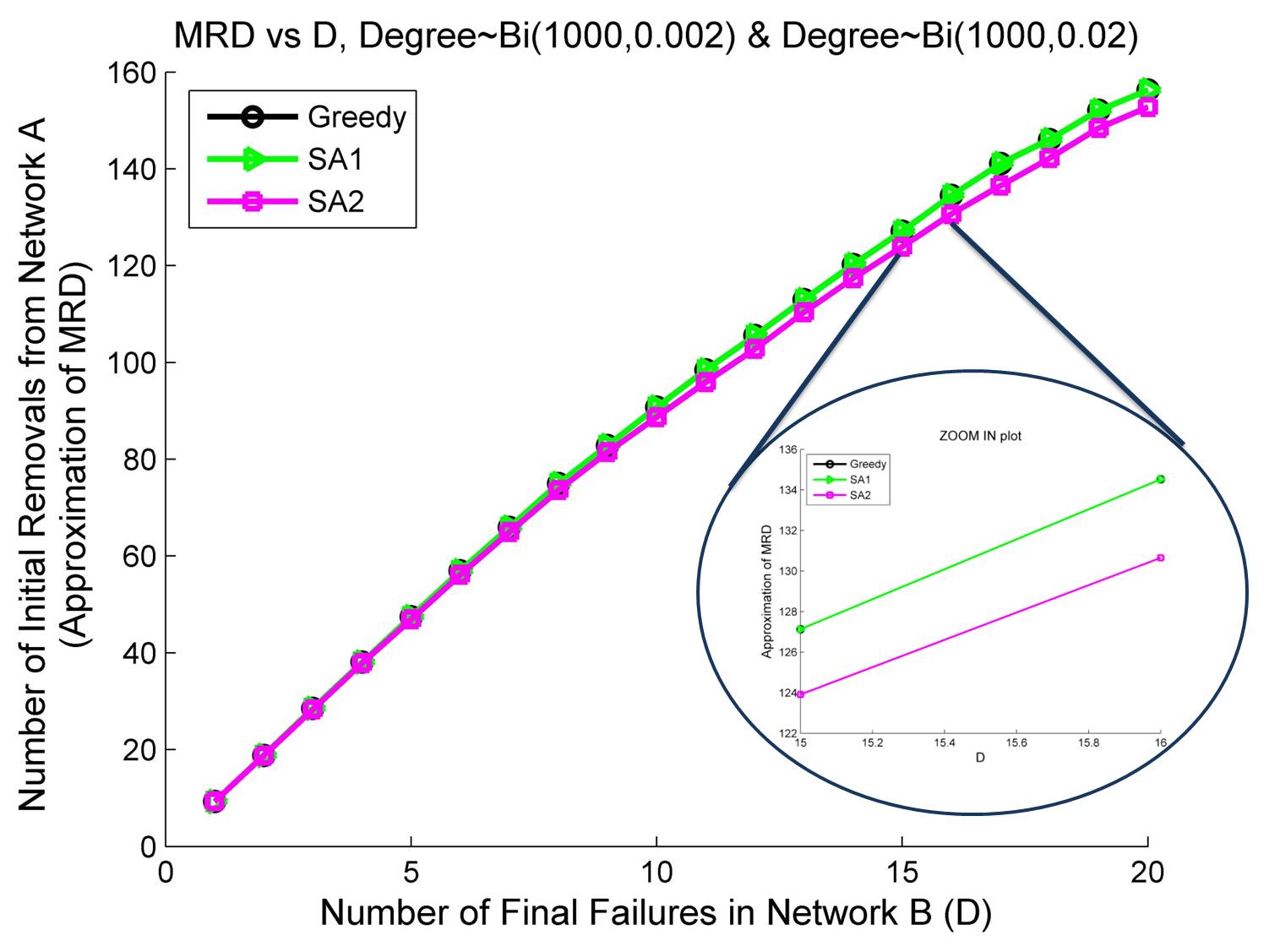}\vspace{-0.3cm}}
\subfigure[Run-Time vs Final Failure Size]
{\label{Time_LargeNetwork_TwoBinomial}\includegraphics[scale=0.05]{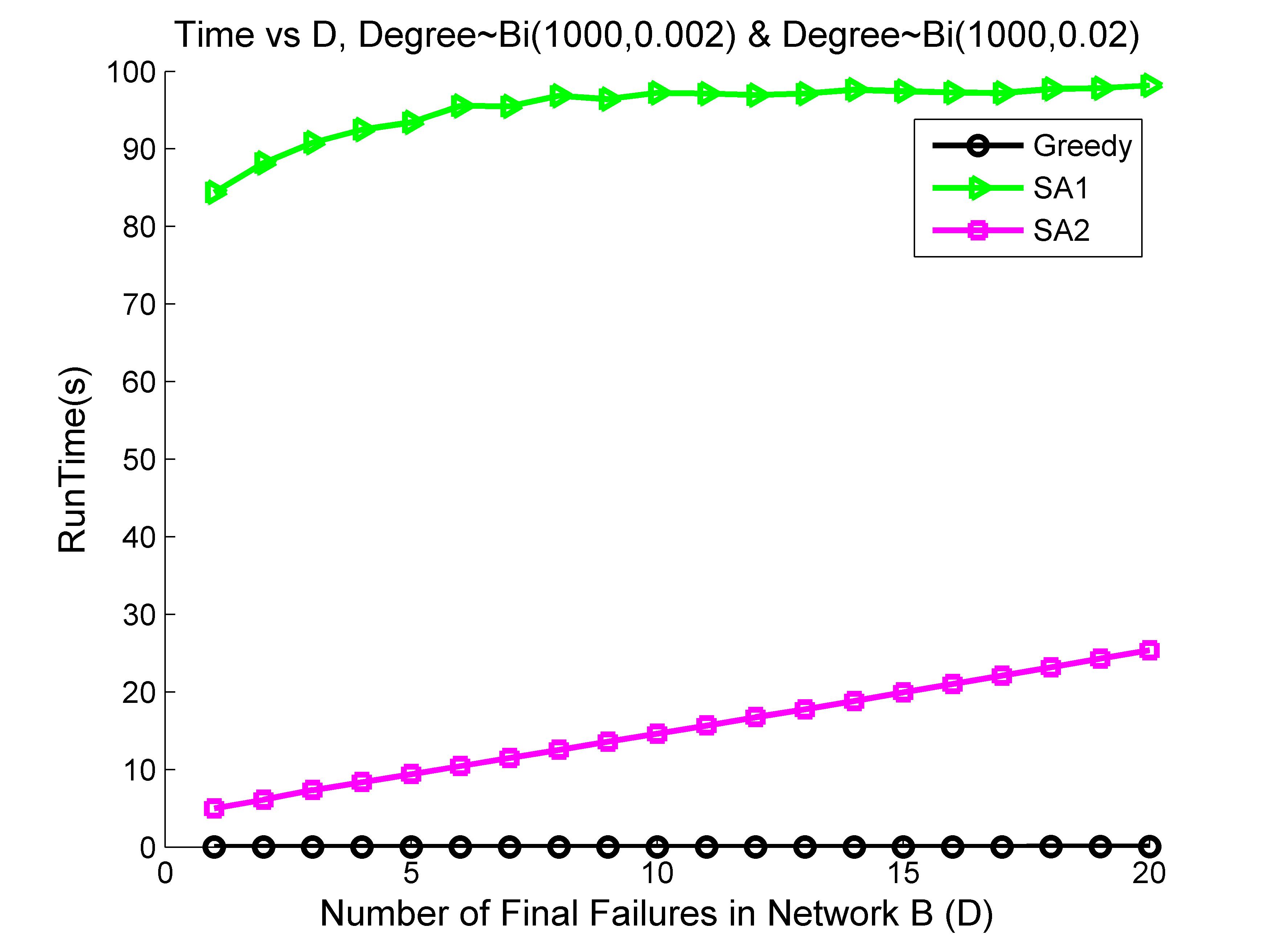}\vspace{-0.3cm}}
\caption{Minimum Node Removal and Run-Time vs Final Failure Size, Type(2) network of size $N=1000$ and $k_1,k_2=[2,20]$, Failure sizes $D \in [1,\cdots,20]$, Two Binomial Distribution}
\label{LargeNetwork_TwoBinomial}\vspace{-0.2cm}
\end{figure}

%%%%%%%%%%%%%%%%%%%%%%%%%%%%%%%%%%%%%%%%%%%%%%%%%%%%%%%%%%%%%%%%%%%%%%%%%%%%%%%%%%%%%%%%%%%%%%%%%%%%%%%%%%%
\section{Robust Design}\label{Design_sec}
In this section, our goal is to design the interdependency between two given networks $A$ and $B$ with star topologies such that network $B$ is robust to failures in network $A$ and network $A$ is robust to failures in network $B$. For simplicity, we assume that both networks have the same number of nodes $N_A=N_B=N$. We also assume that the number of edges between the networks is $E$.

We introduce two definitions for robustness, and propose algorithms for robust design under each definition.

\begin{definition}\label{Lexi_Robust}
\textbf{Lexicographic Robustness:} Network $G^*$ is $K-$robust if for every $D \in \{1, \cdots, K\} $, it has the largest $\mathcal{MR}(D)$ among all networks $G$ with the same number of nodes and edges.
\end{definition}

\begin{definition}\label{Relative_Robust}
\textbf{Relative Robustness:} Network $G^*$ is the most robust network if it has the largest lower bound on the relative $\mathcal{MR}(D)$ for all values of $D \in \{1, \cdots, N\}$; i.e. largest $\min_{1\leq D\leq N}\frac{\mathcal{MR}(D)}{D}$.
\end{definition}

\subsection{Design Under Lexicographic Definition}

\begin{proposition}\label{Robust_Regular}
Consider the set of bidirectional interdependent networks with $N$ nodes and $kN$ edges. The $k$-regular network is the 1-robust network.
\end{proposition}
\begin{proof}
By contradiction - By Proposition \ref{Polynomial_MRD}, in a $k$-regular network, $\mathcal{MR}(D=1)=k$. Suppose that the 1-robust network is irregular. This means that there exist at least one node with degree less than $k$; thus, $\mathcal{MR}(D=1)<k$ which is a contradiction.
\end{proof}

Note that for arbitrary values of $E$, the 1-robust network contains a $k$-regular subnetwork with $k=\lfloor \frac{E}{N} \rfloor$ .

Next, we want to design a 2-robust network. By definition, a 2-robust network is 1-robust, as well. Thus, it is a regular graph by Proposition \ref{Robust_Regular}. However, it can be seen from Figure \ref{Diff_Robustness} that not all regular graphs have the same 2-robustness. In particular, it can be seen that the minimum number of node removals from network $A$ ($B$) to cause the failure of any two nodes in network $B$ ($A$) is $3$. However, in network $G_2$ this value is $2$. Comparison of the structures of graphs $G_1$ and $G_2$ shows that $G_2$ has a more clustered structure than $G_1$. Our goal is to find the structure of the most 2-robust network.

\begin{figure}[ht]
\centering
\subfigure[Graph $G_1$]
{\label{Diff_Robustness1}\includegraphics[scale=0.2]{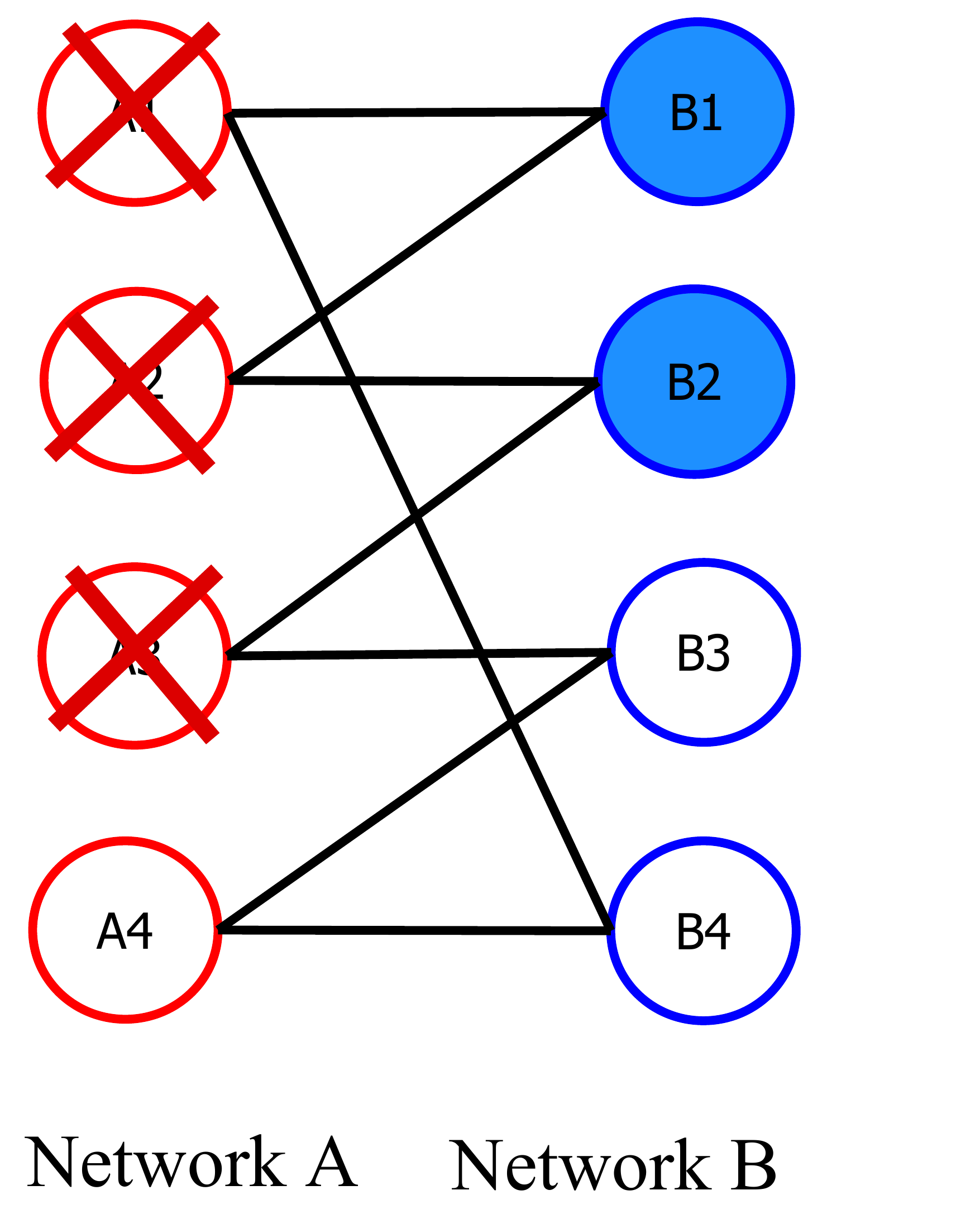}}                
\subfigure[Graph $G_2$]
{\label{Diff_Robustness2}\includegraphics[scale=0.2]{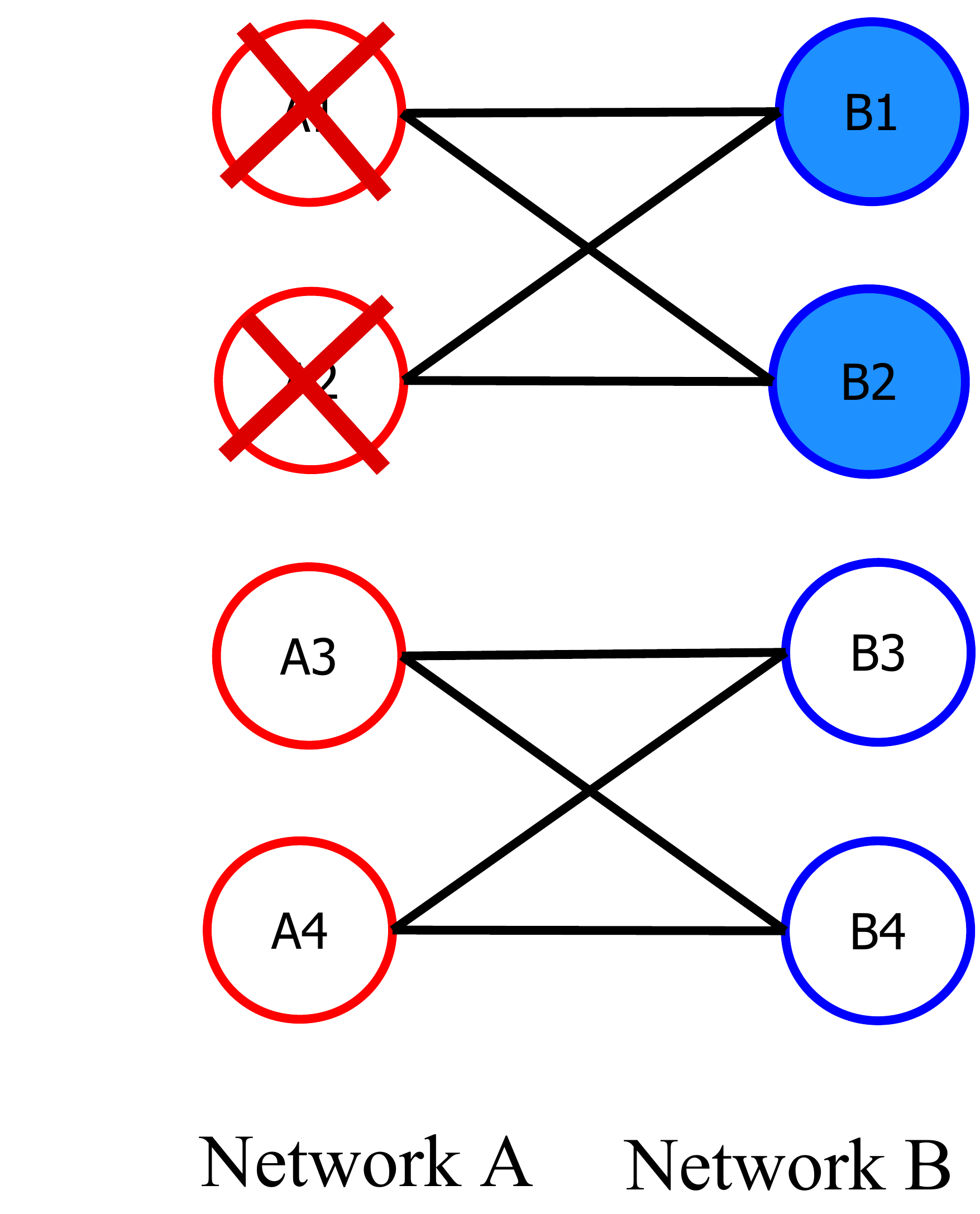}}
\caption{Regular Interdependent networks with robustness for two failures. In network $G_1$, minimum node removal to cause the failure of any two nodes is $3$, where in network $G_2$, the minimum node removals is $2$.}
\label{Diff_Robustness}
\end{figure}

In order to find the structural pattern of the 2-robust networks, we formulate the optimal design problem as an ILP (See Appendix \ref{Robust_ILP} for the formulation details). Figure \ref{robust_simulation} shows the pattern of $\mathcal{MR}(D)$ for networks with different network sizes and node degrees. It can be seen that for any give degree $k$, as number of nodes $N$ increases, $\mathcal{MR}(D)$ increases until it reaches a threshold. This observation is summarized as follows. 

\begin{figure}[ht]
	\centering
	\includegraphics[scale=0.28]{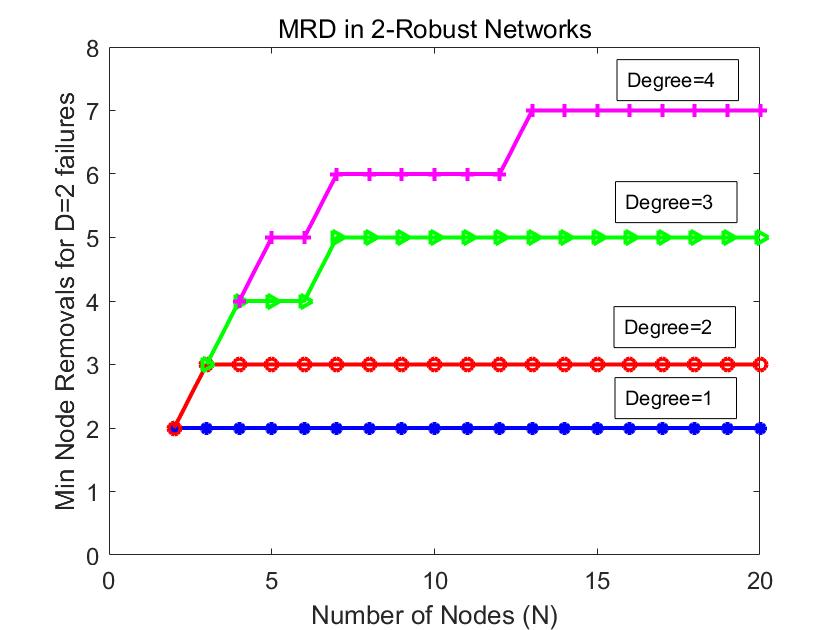}
	\caption{$\mathcal{MR}(D)$ for 2-robust networks with different sizes and node-degrees.}
	\label{robust_simulation}
	\vspace{-0.2cm}
\end{figure}

\textit{Conjecture:} Let  $N_0=k(k-1)+1$ where $k\geq2$. Any 2-robust network with $N \geq N_0$ nodes and degree $k$ has $\mathcal{MR}(D=2)=2k-1$. Moreover, any 2-robust network with $N < N_0$ nodes and degree $k$ has $\mathcal{MR}(D=2)<2k-1$.

In the following, we prove the correctness of this conjecture.

\begin{lemma}\label{2robust_Lemma1}
For $N = N_0$ and $k\geq2$, one can construct a 2-robust network with $N$ nodes and degree $k$ such that $\mathcal{MR}(D=2)=2k-1$.
\end{lemma}
\begin{proof}
By construction - Divide the nodes in network $A$ and $B$ into four groups as in Figure \ref{Robust_Construction}. From left, group 1 has a single node from network $A$, group 2 has $K$ nodes from network $B$, group 3 has $k$ batches of $k-1$ nodes from network $A$, and finally group 4 has $k-1$ batches of $k-1$ nodes from network $B$. Note that the total number of nodes in both networks $A$ and $B$ is $k(k-1)+1$.

Next, we connect the nodes as follows. Connect the single node in group 1 to all the $k$ nodes in group 2. Next, connect node $i$ in group 2 to all the $k-1$ nodes inside batch $i$ in group 3. Finally, connect node $i$ from the last batch of group 3 to all the $k-1$ nodes of batch $i$ in group 4. So far, all the nodes in group 1, group 2 and the last batch of group 3 has degree $k$. Moreover, every pair of nodes in group 2 (part of network $B$) share exactly one neighbor which is the single node in group 1. Next, we connect the nodes in the batches $1,\cdots,(k-1)$ in group 3 to batches $1,\cdots,(k-1)$ in group 4 as follows.

\textit {For $i,j \in \{1,\cdots,k-1\}$, connect node $i$ in batch $j$ of group 3 to node $i\pmod{k-1}$ in batch $1$ of group 4, node $i+j-1\pmod{k-1}$ in batch $2$ of group 4, ... , node $i+j+k-3\pmod{k-1}$ in batch $k-1$ of group 4}

This rule satisfies the following conditions: 
\begin{enumerate}
	\item every node is connected to $k-1$ new edges;
	\item every node from group 3 is connected to exactly one node inside each batch in group 4;
	\item no pair of nodes inside a batch in group 3 share a neighbor in group 4;
	\item every pair of nodes from two different batches in group 3 share exactly one neighbor in group 4;
	\item every node from group 4 is connected to exactly one node inside each of the first $(k-1)$ batches of group 3;
	\item no pair of nodes inside a batch in group 4 share a neighbor in the first $(k-1)$ batches of group 3;
	\item every pair of nodes from two different batches in group 4 share exactly one neighbor in the first $(k-1)$ batches of group 3. 
\end{enumerate}

Note that conditions (2)-(4) guarantee that every pair of nodes in network $A$ share exactly one neighbor. In addition, conditions (5)-(7) guarantee that any pair of nodes inside group 4 or between groups 2 and 4 share exactly one neighbor.

Therefore, we have constructed a graph with $N_0$ nodes where every node has degree $k$, and every pair of nodes in either network $A$ or network $B$ share exactly one neighbor; i.e. $\mathcal{MR}(D)=2k-1$.

\begin{figure}[ht]
	\centering
	\includegraphics[scale=0.6]{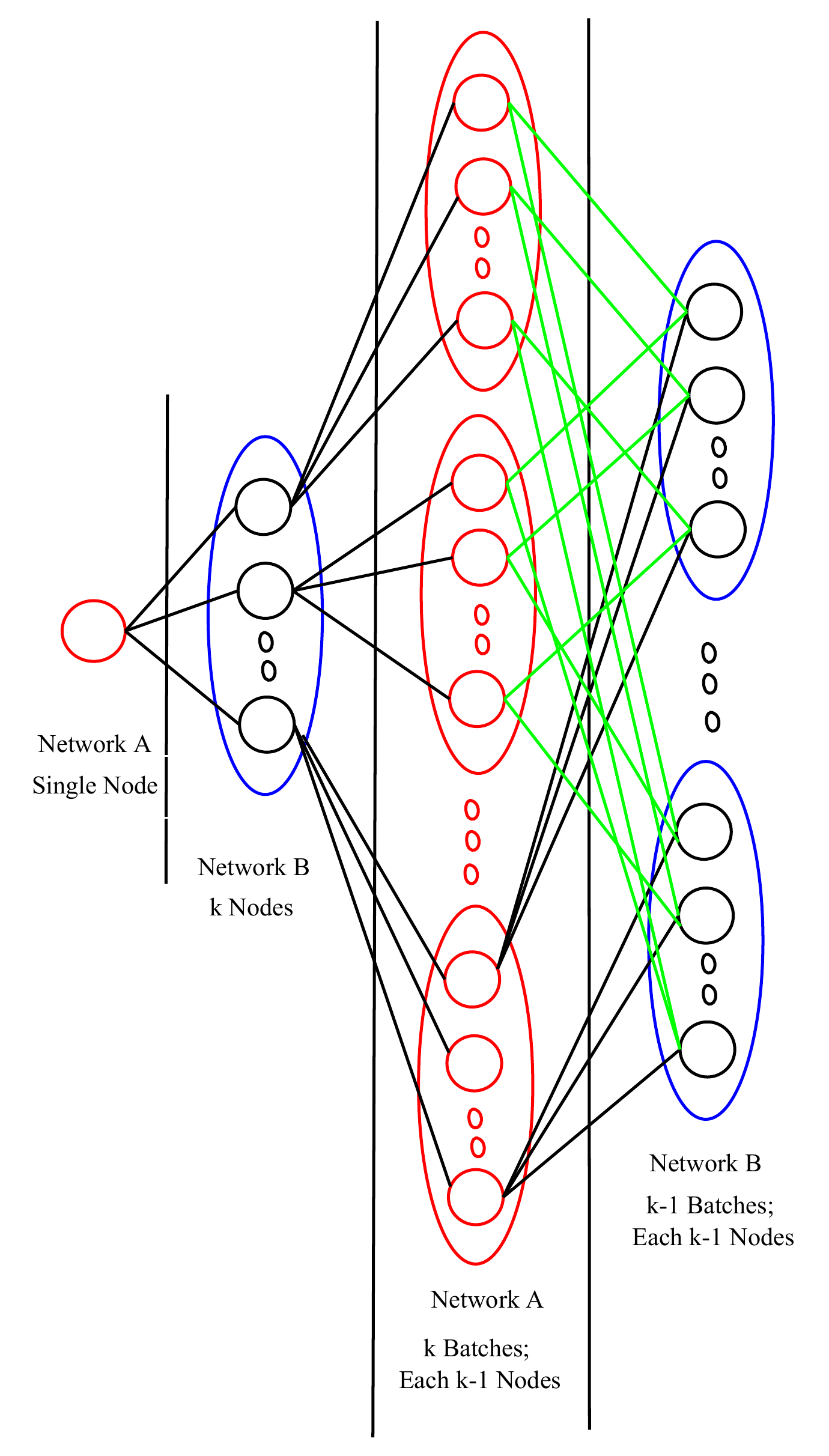}
	\caption{An example of construction of a 2-robust network. Black links denote the first set of edges connecting the nodes in group 1, group 2 and the last batch of group 3. Green links denote the set of edges connecting the nodes in group 3 and group 4.}
	\label{Robust_Construction}
	\vspace{-0.2cm}
\end{figure}

\end{proof}

\begin{lemma}\label{2robust_Lemma2}
For $N \geq N_0$ and $k\geq2$, there exists no regular network with $N$ nodes and degree $k$ such that $\mathcal{MR}(D=2)>2k-1$.
\end{lemma}
\begin{proof}
Since each node has degree $k$, no two nodes can be connected to more than $2k$ nodes; i.e. $\mathcal{MR}(2)\leq 2k$. Next, suppose that there exists a $k-$regular network with $\mathcal{MR}(D=2)=2k$. Thus, every pair of nodes in network $B$ are exactly connected to $2k$ nodes; i.e. each node in network $B$ is connected to $k$ distinct nodes. Equivalently, no node in network $A$ is connected to two nodes in network $B$; i.e., every node in network $A$ has degree 1, which is a contradiction to $k \geq 2$. 
\end{proof}

\begin{lemma}\label{2robust_Lemma3}
For $N < N_0$ and $k\geq2$, there exists no regular network with $N$ nodes and degree $k$ such that $\mathcal{MR}(D=2) \geq 2k-1$.
\end{lemma}
\begin{proof}
Suppose $N=N_0-1=k(k-1)$. Consider an arbitrary node $i$ from network $A$. Node $i$ is connected to $k$ nodes in network $B$, namely set $X$, where each of these nodes are also connected to $k-1$ nodes in network $A$, namely set $Y$. Note that $i \notin Y$. Since the total number of nodes in each network is $k(k-1)$, $|Y|\leq k(k-1)-1$. Thus, there exist at least two nodes in $X$ that share a neighbor in $Y$. On the other hand, all nodes in $X$ share node $i$ as their neighbor, too. Therefore, there exist at least two nodes in network $B$ that share more than one node in network $A$. The same argument holds for network $A$. Thus, $\mathcal{MR}(D=2) < 2k-1$ for $N=N_0-1$. Clearly, as the total number of nodes decreases, $\mathcal{MR}(D=2)$ decreases, too. Thus, for any regular network with $N<N_0$, $\mathcal{MR}(D=2) < 2k-1$.
\end{proof}

In order to design a 3-robust network under the definition of Lexicographic robustness, one should search among the 2-robust networks which becomes very complicated. In the next section, we consider the relative robustness and show its relation with expander graphs.

\subsection{Design under Relative Robustness}
First, we show that by definition, a network with large relative $\mathcal{MR}(D)$ has also a large node expansion. Then, we show the construction of expander graphs.

\begin{definition}
The Node Expansion of a bipartite graph $G=\{A,B\}$, denoted by $h(G)$, is defined as:\\
	\begin{equation}
		h(G)=\min_{S \subseteq B} \frac{|\mathcal{N}(S)|}{|S|}
	\end{equation}
	where $\mathcal{N}(S)$ denotes the neighbor nodes of set $S$.
\end{definition}

\begin{lemma}\label{expansion_equivalency}
Under relative robustness definition, the most robust network has the largest node expansion.
\end{lemma}

\begin{proof}
\begin{subequations}
	\begin{align}
	 & \min_{1 \leq D \leq N} \frac{\mathcal{MR}(D)}{D} = \\
	 & \min_{1 \leq D \leq N} \frac{\min_{Y_D \subset B , |Y_D|=D} |\mathcal{N}(Y_D)|}{D} = \\
	 & \min_{S \subseteq B} \frac{|\mathcal{N}(S)|}{|S|} = h(G)
	\end{align}
\end{subequations}
\end{proof}

Lemma \ref{expansion_equivalency} shows that in order to design a robust interdependent network, it is enough to design a network with large node expansion; i.e. an expander bipartite graph. Analysis and design of node and edge expander graphs is a well-studied topic (See \cite{hoory2006expander, vadhan2012pseudorandomness}). It has been shown that some random graphs share the properties of expander graphs, and they have been used to prove the existence of expander graphs. However, explicit construction of an expander graph is more difficult and there exist only three main strategies for designing them (See \cite{yehudayoff2012proving,Expanders2013Odonnell} for more details). 

In the following, we mention one of the main results regarding random graphs and their relation with expander graphs.

Let a bipartite graph $G=\{A,B\}$ be a $(D,r)$ node expander if for all sets $S \subseteq B$ of size at most $D$, the neighborhood $\mathcal{N}(S)$ is of size at least $r|S|$. Moreover, let $Bip_{N,k}$ be the set of bipartite graphs that have $N$ nodes on each side and are $k$-Bregular, meaning that every node in network $B$ has degree $k$.

\begin{theorem}\label{random_expander}
For every constant $k$, there exists a constant $\alpha > 0$ such that for all $N$, a uniformly random graph from $Bip_{N,k}$ is an $(\alpha N,k - 2)$ node expander with probability at least $\frac{1}{2}$.
\end{theorem}
\begin{proof}
See \cite{vadhan2012pseudorandomness} for proof.
\end{proof}

Note that for every $S \subseteq B$, the largest possible neighbor has size of $k|S|$, and Theorem \ref{random_expander} denotes that in a uniform random graph, every $S \subseteq B$ of size $\alpha N$ has neighbors of size $(k-2)|S|$ with probability more than half.

%%%%%%%%%%%%%%%%%%%%%%%%%%%%%%%%%%%%%%%%%%%%%%%%%%%%%%%%%%%%%
\section{Discussion}\label{Discussion_sec}

Throughout this paper, we analyzed the robustness of interdependent networks for star topologies. We defined metric $\mathcal{MR}(D)$ and proved the hardness of evaluating this metric for arbitrary values of $D$ in both unidirectional and bidirectional interdependent networks. We also proved that uniform distribution of edges in a network would result in more robustness. A natural direction of future research would be analyzing more general topologies for interdependent networks which makes the problem more complicated. In fact, we prove that for the tree topologies, evaluating metric $\mathcal{MRB}(D)$ in bidirectional interdependent networks becomes NP-complete even for the special case of total failure in network $B$; i.e. $D=N_2$. This is due to the fact that failures cascade both inside the networks and between the networks.

\begin{theorem}\label{Tree_Hardness}
For $D=N_2$, finding the $\mathcal{MRB}(D)$ in a bidirectional interdependent network with tree topology is an NP-complete problem.
\end{theorem}
\begin{proof}
The proof is based on a reduction from the problem of finding $\mathcal{MRB}(D)$ in a unidirectional interdependent network with star topology which is proved to be NP-complete \cite{parandehgheibi2013robustness}. Consider graph $G$, an arbitrary unidirectional interdependent network with star topology, with nodes $\{A_1,\cdots,A_{N_1}\}$ in network $A$, nodes $\{B_1,\cdots,B_{N_2}\}$ in network $B$ and two sources $S_1$ and $S_2$ which are directly connected to nodes in network $A$ and $B$, respectively. Let $E_{AB}$ represent the set of edges from network $A$ to network $B$. Similarly, let $E_{BA}$ represent the set of edges from network $B$ to network $A$.

Construct graph $G'$, a bidirectional interdependent network with tree topology, using graph $G$ as follows. Generate graph $G'$ similar to graph $G$, and remove all the edges $E_{AB}$ and $E_{BA}$. For each node $A_i$, generate a child node $A_{ii}$ with an edge from $A_i$ to $A_{ii}$. Similarly, for each node $B_j$, generate a child node $B_{jj}$ with an edge from $B_j$ to $B_{jj}$. For every edge in $E_{AB}$ connecting node $A_i$ to node $B_j$, connect the child node $A_{ii}$ to node $B_j$ . Similarly, for every edge in $E_{BA}$ connecting node $B_j$ to node $A_i$, connect the child node $B_{jj}$ to node $A_i$ (See Figure \ref{TreeHardnessProof}). This construction guarantees that removal of any set of parent nodes $X$ in both graphs $G$ and $G'$ would lead to the failure of the same parent nodes $Y$ in both graphs. Moreover, note that under this construction, failure of any parent node in graph $G'$ guarantees failure of its child.

\begin{figure}
\centering
\subfigure[Graph $G$]
{\label{UniStarProof}\includegraphics[scale=0.3]{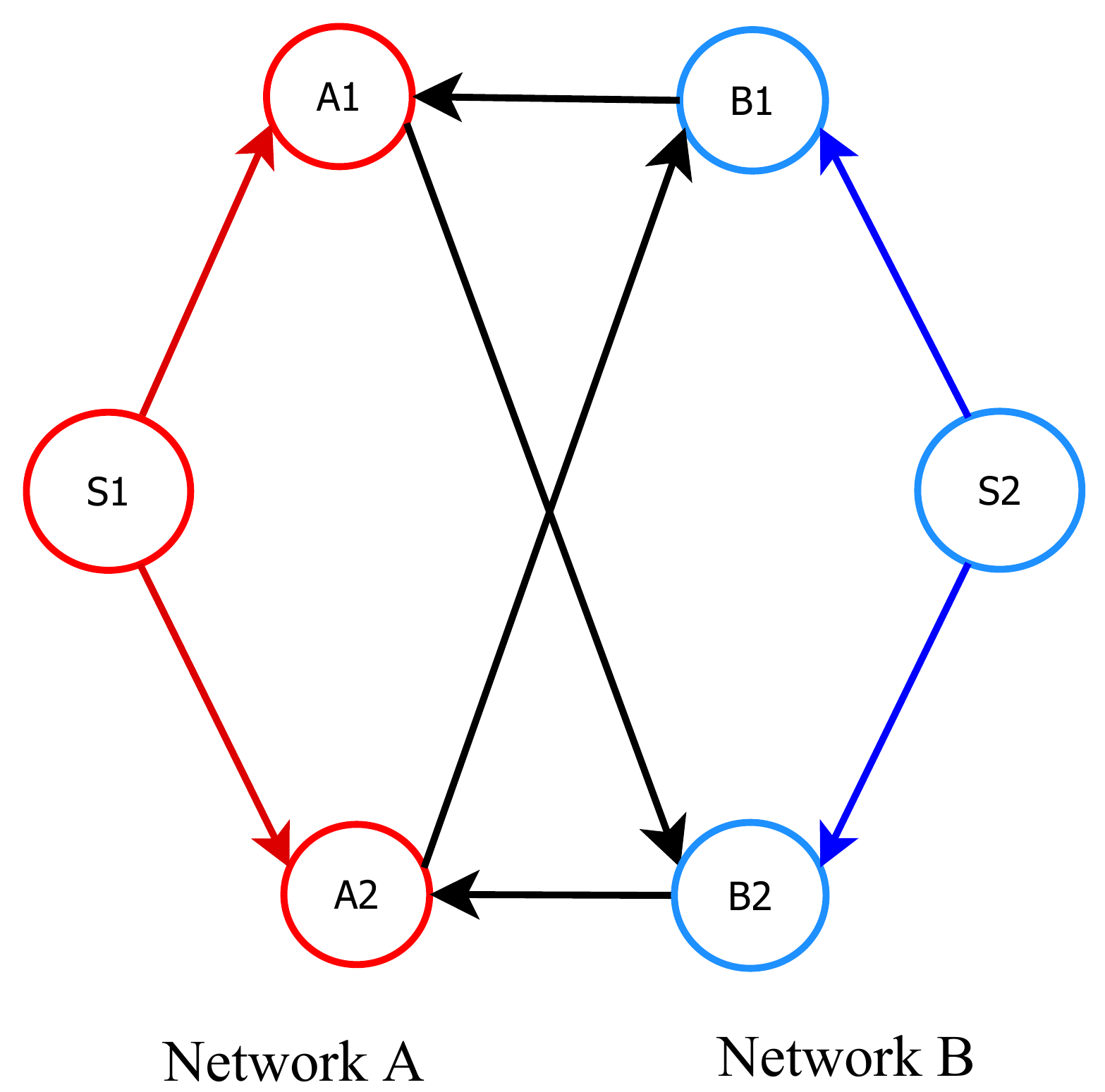}}                
\subfigure[Graph $G'$]
{\label{BiTreeProof}\includegraphics[scale=0.3]{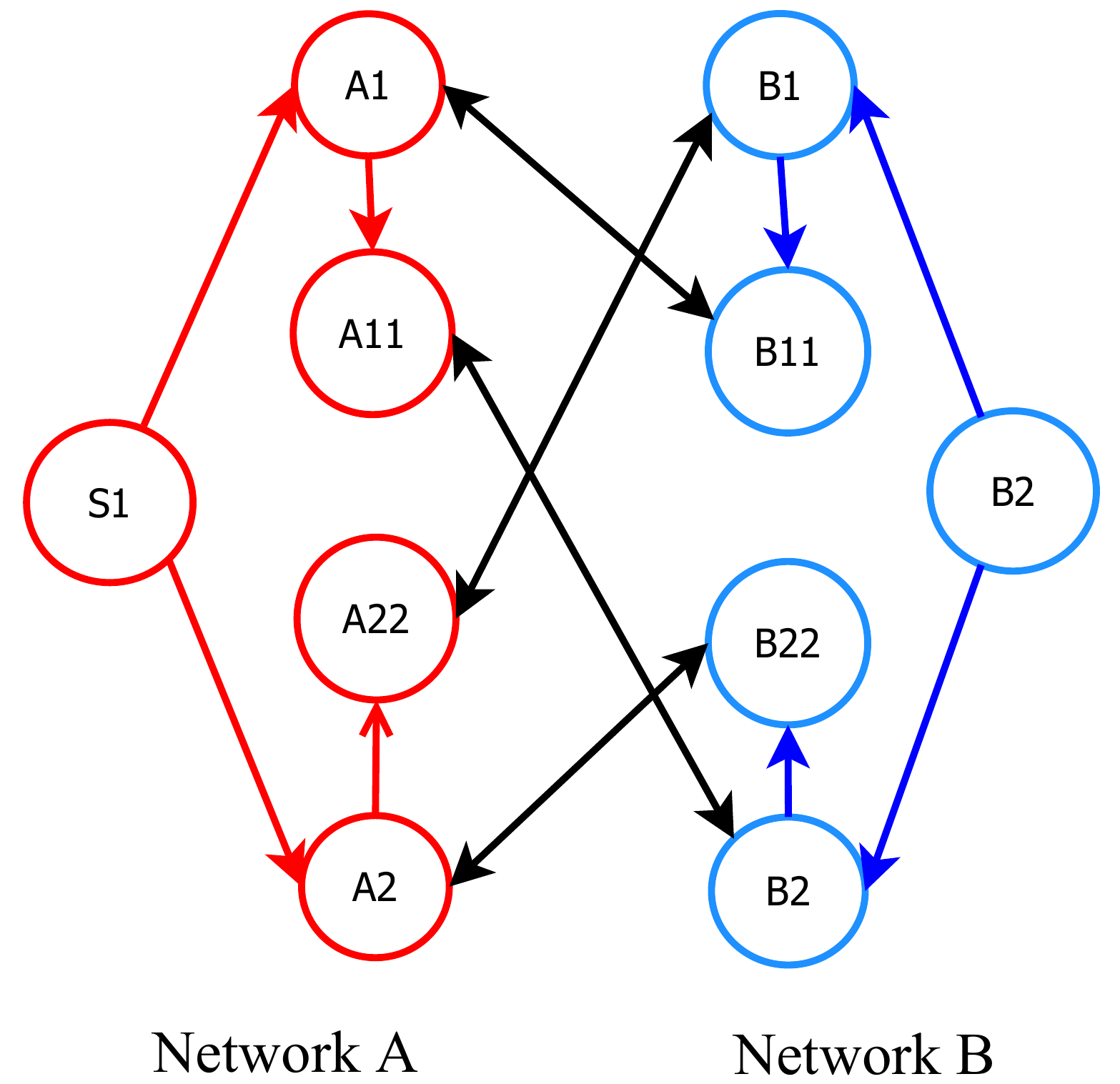}}
\caption{Graph Topologies for Proof of Theorem \ref{Tree_Hardness}}
\label{TreeHardnessProof}
\end{figure}

Next, we show that if $\mathcal{MRB}(D=N'_2)$ in graph $G'$ can be found in polynomial time, $\mathcal{MRB}(D=N_2)$ in graph $G$ can also be found in polynomial time which is a contradiction \cite{parandehgheibi2013robustness}.

Suppose $R'$ is the optimal set of node removals that lead to the failure of the entire network $B$ in graph $G'$. Note that $R'$ could contain both parent and child nodes; however, it is clear that the effect of removal of any parent node $A_i$ (respectively, parent node $B_j$) is more than or equal to the effect of removal of its child node $A_{ii}$ (respectively, child node $B_{jj}$). Thus, we replace all the child nodes in $R'$ to the parent nodes, and call the new set $R'_P$. It is enough to show that $R'_P$ is also the optimal removal set in graph $G$.

Due to the construction of graph $G'$ from $G$, removal of $R'_P$ in $G$ leads to the total failure of network $B$. Next, we prove by contradiction that it is also the optimal solution. Suppose that $R$ is the optimal removal in $G$ where $|R|<|R'_P|$. By construction, removal of nodes $R$ in $G'$ will lead to the total failure of network $B$; thus, $R'_P$ is not optimal which is a contradiction.

\end{proof}

In Theorem \ref{Tree_Hardness}, we proved that finding the optimal removal sets in graph $G$ and $G'$ are equivalent. This was due to the fact that under our construction, parent nodes could replicate the entire cascading failure process. However, this is not true for any arbitrary bidirectional interdependent network with tree topology. Consider the following example.

\textit{Example-} Consider the network in Figure \ref{CounterExample}. Here, nodes $A_1$, $A_2$, $A_3$ and $B_1$ are the parent nodes directly connected to the sources, and the rest of nodes are children. Moreover, the nodes between the two networks are connected via bidirectional edges. Suppose parent node $A_2$ fails. Thus, child nodes $A_{21}$, $B_{21}$ and $B_{22}$ fail which leads to the failure of parent node $A_3$.

\begin{figure}[ht]
	\centering
	\includegraphics[scale=0.3]{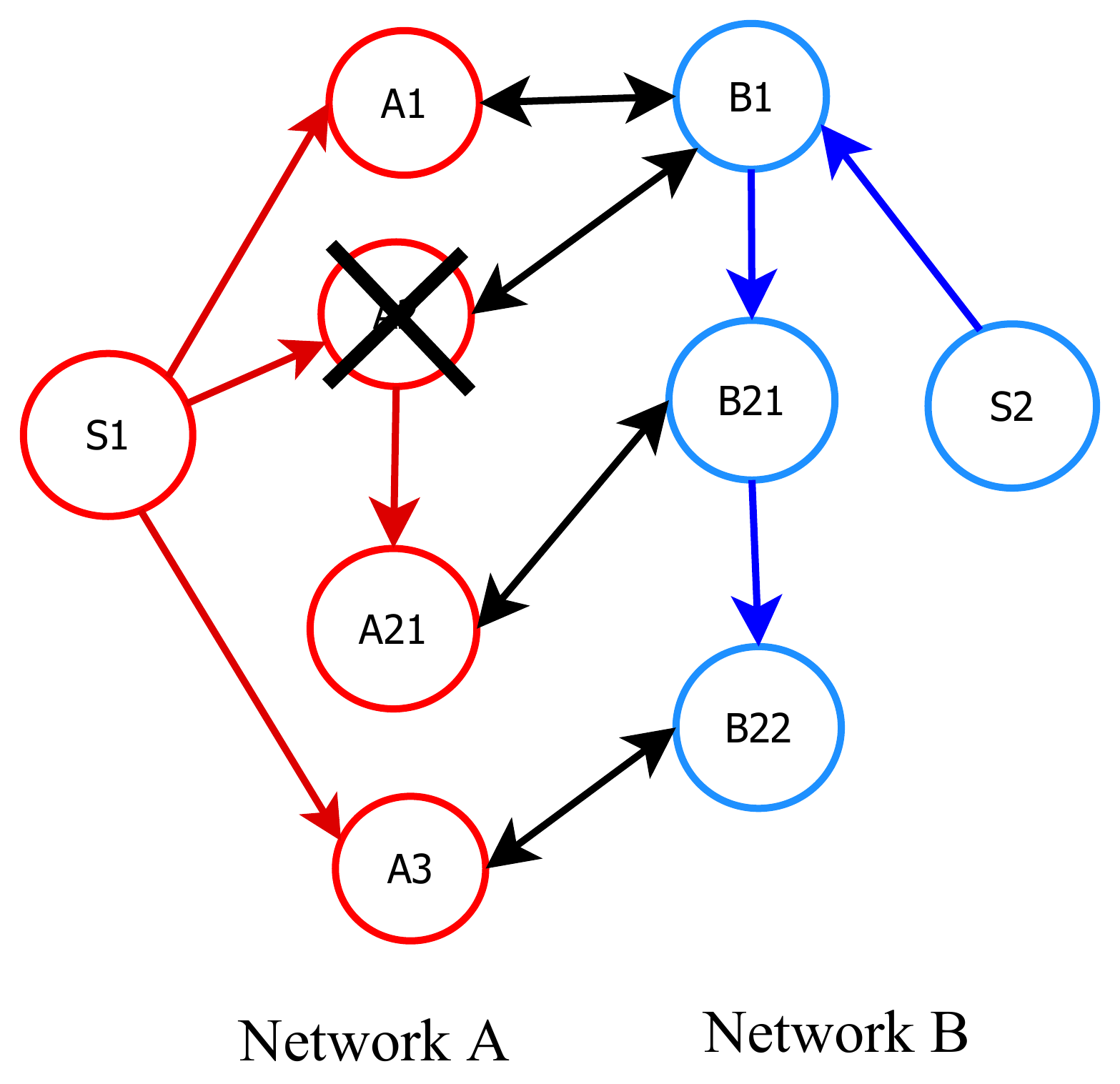}
	\caption{An example of bidirectional interdependent networks with tree topology. Here, the failure of parent node $A_2$ leads to the failure of parent node $A_3$ without affecting any parent node in network $B$.}
	\label{CounterExample}
	\vspace{-0.2cm}
\end{figure}

This example illustrates that the failure of a parent node in network $A$ can lead to the failure of another parent node in network $A$ without affecting any parent node in network $B$. Therefore, there are no graph structural mapping to replicate the cascade of failures from bidirectional networks to unidirectional ones with star topology. Thus, despite the fact that star topologies illustrate many important properties of interdependent networks, the analysis of general topologies cannot be a direct extension of star topologies and requires extensive analysis.

%%%%%%%%%%%%%%%%%%%%%%%%%%%%%%%%%%%%%%%%%%%%%%%%%%%%%%%%%%%%%
\section{Conclusion}\label{Conclusion_sec}
In this paper, we studied the robustness of interdependent networks. We proposed a deterministic model for analyzing interdependent networks with given topologies. We focused on the networks with star topologies to capture the effect of cascading failures due to interdependency. We defined two metrics $\mathcal{MR}(D)$ ($\mathcal{MRB}(D)$) as the minimum number of nodes that should be removed from network $A$ (both networks) to cause the failure of $D$ nodes in network $B$. We proved that evaluating these metrics is not only NP-complete, but also inapproximable. Moreover, we proposed several algorithms based on greedy, randomize rounding and simulated annealing for evaluating our metrics and compared their performances using simulation results. We proved that in the networks with the same number of nodes and edges, those with bidirectional interdependency are more robust than the ones with unidirectional interdependency. Next, we introduced two closely related definitions for robust interdependent networks, proposed algorithms for explicit design, and showed the relation of robust interdependent networks with expander graphs. Finally, we discussed some ideas about analysis of interdependent networks with general topologies.

\bibliographystyle{IEEEtran}
\bibliography{reference}

%%%%%%%%%%%%%%%%%%%%%%%%%%%%%%%%%%%%%%%%%%%%%%%%%%%%%%%%%%%%%%%%%%%%%%%%%%%%%%%%%%
%%%%%%%%%%%%%%%%%%%%%%%%%%%%%%%%%%%%%%%%%%%%%%%%%%%%%%%%%%%%%%%%%%%%%%%%%%%%%%%%%%
%%%%%%%%%%%%%%%%%%%%%%%%%%%%%%%%%%%%%%%%%%%%%%%%%%%%%%%%%%%%%%%%%%%%%%%%%%%%%%%%%%

\appendices

\section{Proof of Theorem \ref{MRD_hardness}}\label{MRD-hardness-proof}
Consider graph $G$ as a bidirectional interdependent network. According to Lemma \ref{NeighborRemoval}, nodes in set $Y \in B$ fail if \textit{all} of their direct neighbors, namely nodes in $X \in A$, are removed. Note that nodes in $X$ can have direct neighbors other than nodes in $Y$; i.e. nodes in set $B \setminus Y$ (See Figure \ref{Proof1}). Finding $\mathcal{MR}(D)$ in graph $G$ is equivalent to finding the smallest set $X$ whose removal leads to the failure of set $Y$ with at least $D$ nodes.

In order to prove the hardness of finding $\mathcal{MR}(D)$, we construct a new bipartite graph $G'$ as the complement of graph $G$ where all of the interdependent edges are removed, and all disjoint pairs are connected with a bidirectional edge (See Figure \ref{Proof2}). Since there is no connection between nodes in sets $Y$ and $A\setminus X$ in graph $G$, subgraph $(Y,A\setminus X)$ forms a complete bipartite graph (biclique) in $G'$. Therefore, finding $\mathcal{MR}(D)$ is equivalent to finding the largest set $A\setminus X$ where $(A\setminus X,Y)$ is a biclique and $Y$ contains at least $D$ nodes.

\begin{figure}[ht]
\centering
\subfigure[Graph $G$]
{\label{Proof1}\includegraphics[scale=0.25]{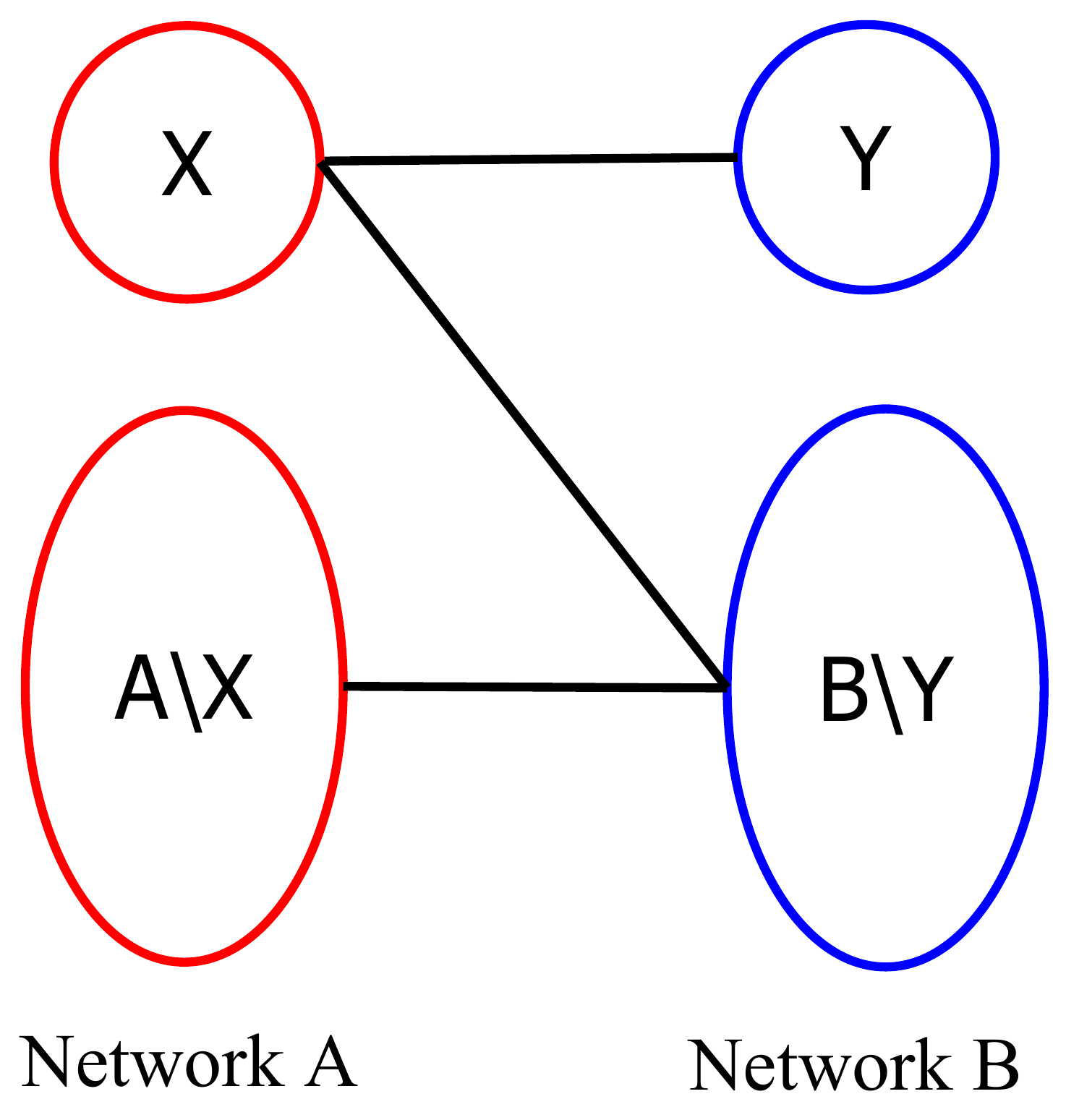}\vspace{-0.3cm}}
\subfigure[Graph $G'$]
{\label{Proof2}\includegraphics[scale=0.25]{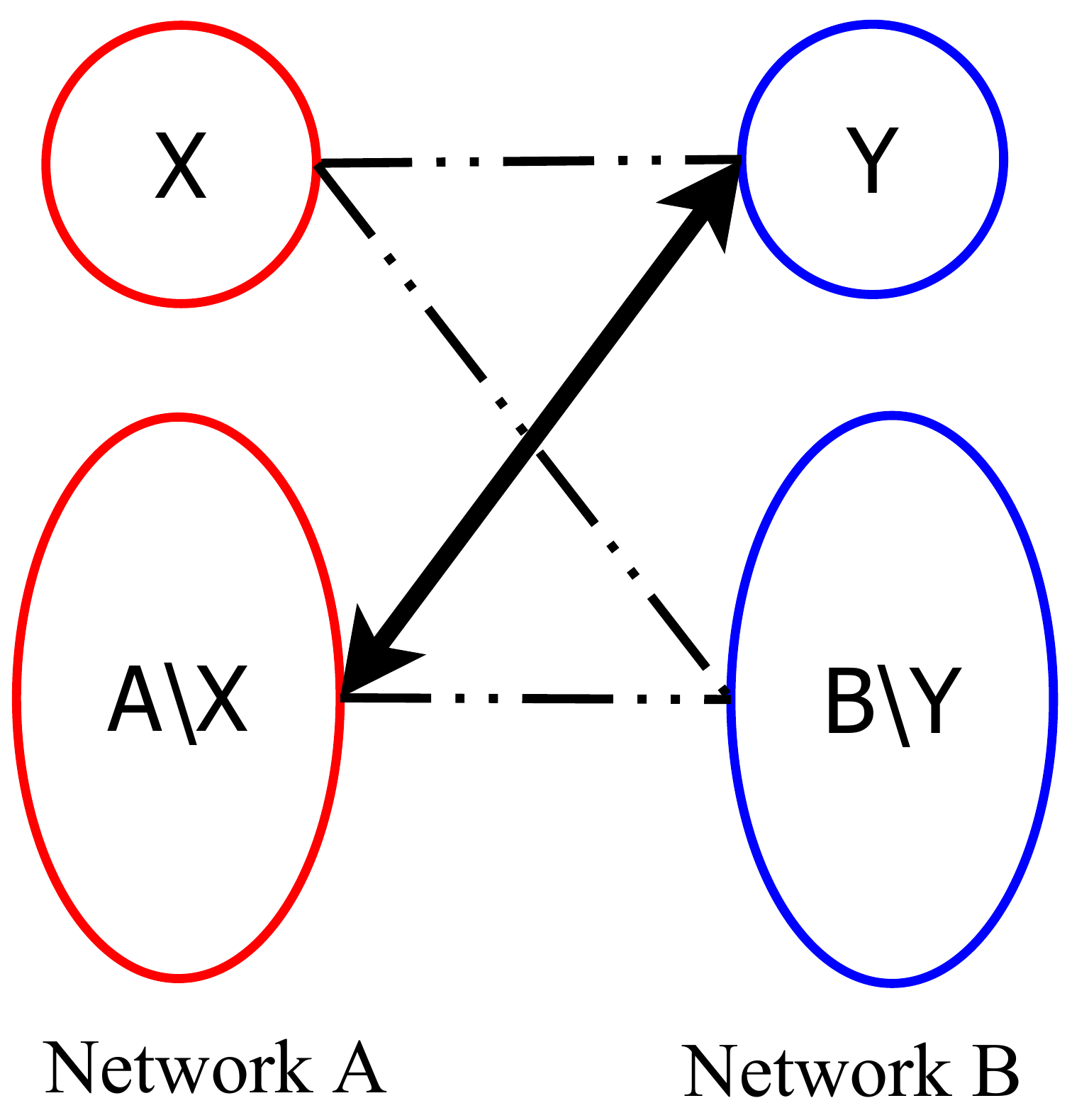}\vspace{-0.3cm}}
\caption{Graph Topologies in Proof of Theorem \ref{MRD_hardness}}
\label{RAID-fig}\vspace{-0.2cm}
\end{figure}

Next we show that finding such biclique is an NP-complete problem. The proof of hardness is based on a reduction from the problem of balanced complete bipartite subgraph which is known to be NP-complete \cite{Garey}. 

\begin{definition}
\textit{Balanced Complete Bipartite Subgraph:} Given a bipartite graph $G=\{V,E\}$ and a positive integer $K \leq |V|$, are there two disjoint subsets $V_1,V_2 \subset V$ such that $|V_1|=|V_2|=K$ and any pair of nodes in $(V_1,V_2)$ be an edges in $E$?
\end{definition}

Next, we show that if we can solve our problem in graph $G$ for every value $D$, then we can solve the Balanced Complete Bipartite Subgraph problem in graph $G'$ for every value $K=D$ as follows. Suppose that $\mathcal{MR}(D)$ can be evaluated in graph $G$ in polynomial time. Thus, for every value of $D$, we can find the largest $A\setminus X$ subgraph of $G'$ in polynomial time so that $(A\setminus X,Y)$ is complete and $|Y|\geq D$. If $|A\setminus X| \geq D$, there exists a complete bipartite graph of size $D$ in graph $G'$, and if $|A\setminus X| < D$, there exists no complete bipartite graph of size $D$ in $G'$. Therefore, we can decide if there exists a balanced complete bipartite subgraph of size $K=D$ in $G'$ in polynomial time which is a contradiction. Thus, our problem is NP-complete.

\section{Proof of Corollary \ref{MRBD_hardness}}\label{MRBD-hardness-proof}

For an arbitrary bipartite network $G$, construct network $G'$ as follows. Replace every node $i \in B$ with a cluster of $W$ nodes, where $W$ is a very large number ($W>N_1+N_2$). For every node $j \in A$ connected to node $i \in B $ in network $G$, connect $j\in A$ in network $G'$ to all $W$ nodes replacing $i\in B$ (See Figure \ref{GraphConversion}). Now it is enough to show that if metric $\mathcal{MRB}(D)$ can be evaluated in polynomial time for network $G'$, metric $\mathcal{MR}(D)$ can also be evaluated in polynomial time for network $G$, which is a contradiction to Theorem \ref{MRD_hardness}.

\begin{figure}[ht]
\centering
\subfigure[Subnetwork $G$]
{\label{GraphConversion1}\includegraphics[scale=0.28]{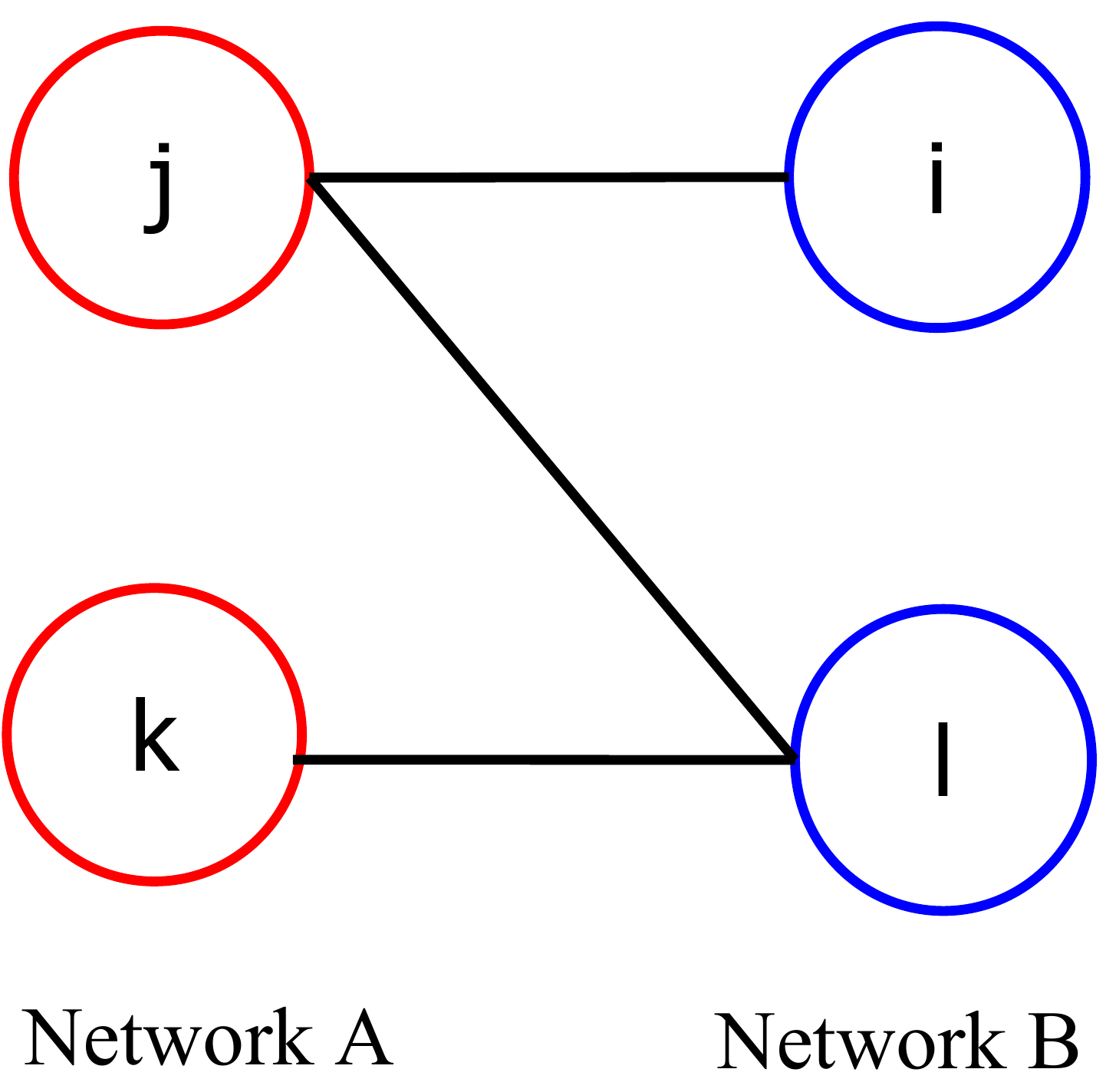}\vspace{-0.3cm}}
\subfigure[Subnetwork $G'$]
{\label{GraphConversion2}\includegraphics[scale=0.28]{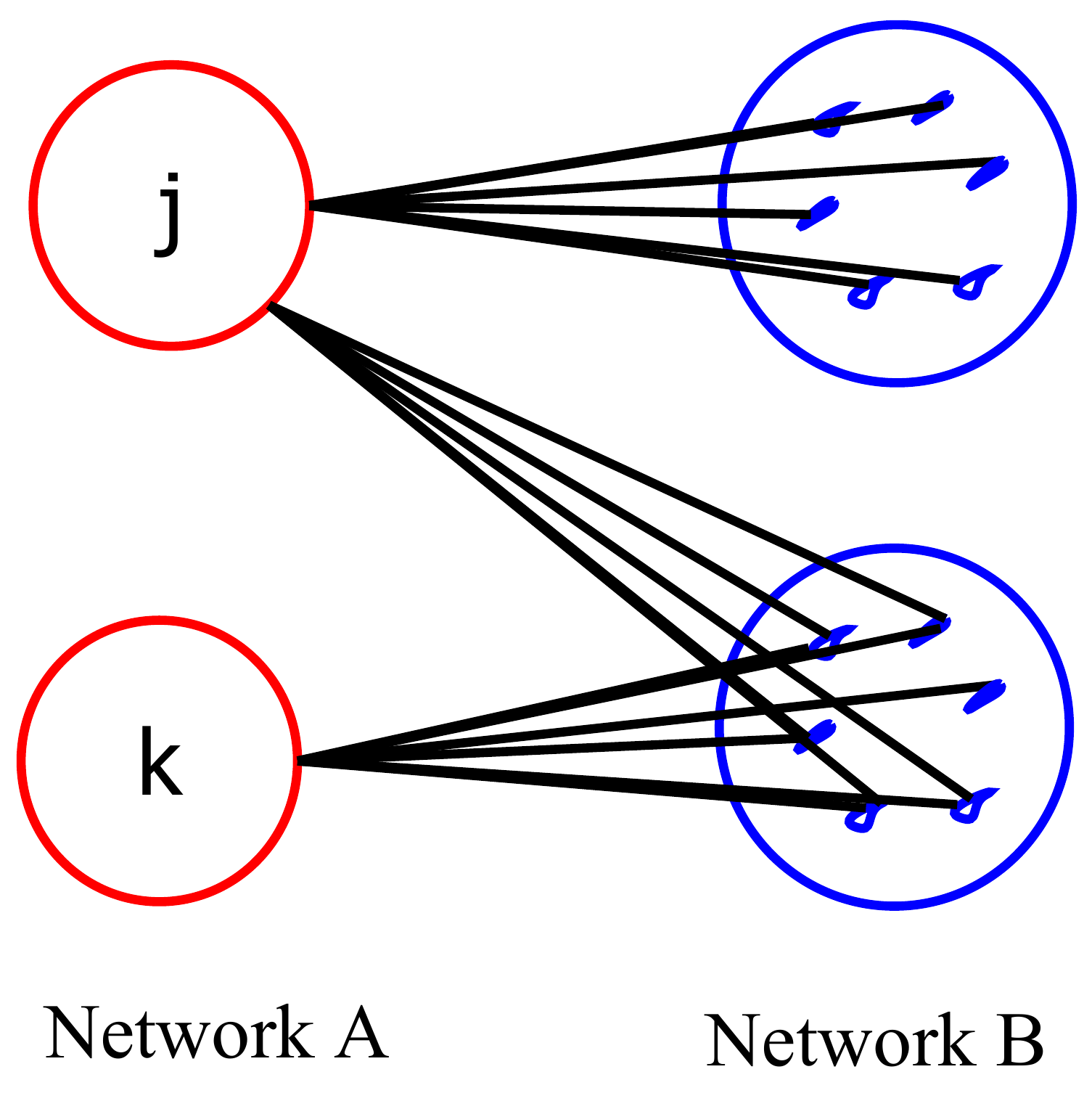}\vspace{-0.3cm}}
\caption{Conversion of subnetwork $G$ to subnetwork $G'$}
\label{GraphConversion}\vspace{-0.2cm}
\end{figure}

Suppose one can evaluate $\mathcal{MRB}(D)$ in network $G'$ for failure of $D'=WD$ nodes in $B$ in polynomial time. It is easy to see that for any removal of nodes $X\in A$ in network $G$ leading to $D$ failures in $B$, removal of the same exact nodes from $A$ in network $G'$ leads to the failure of $D'=DW$ nodes in $B$ and vice versa. This is due to the fact that in network $G'$, all edges between every $j \in A$ and a cluster of $W$ node in $B$ are mapped according to edges in graph $G$. Moreover, since $W>N_1+N_2$, all removals will be only from network $A$. Thus, $\mathcal{MRB}(D)$ for $D'=DW$ failures in network $G'$ is exactly the same as $\mathcal{MR}(D)$ for $D$ failures in network $G$.

\section{Inapproximability of metric $\mathcal{MR}(D)$}\label{Inapproximability_proof}
Consider networks $G$ and $G'$ in Figure \ref{RAID-fig}. For simplicity, we assume that $N_1=N_2=N$. Let $W_D$ be the approximation of $\mathcal{MR}(D)$ in network $G$ where $W_D^*$ is the optimal value. Suppose there exists a PTAS which approximates $\mathcal{MR}(D)$ within factor of $r>1$; i.e. $W_D^* \leq W_D\leq r \cdot W_D^*$. Moreover, define variables $Z_D=N-W_D$ and $Z_D^*=N-W_D^*$.

Moreover, let $X^*$ be the largest balanced biclique in network $G'$. It is easy to see that $X^*=\max_{1\leq D\leq N} \min \{Z_D^*,D\}$. Similarly, let $X$ be the largest balanced biclique in network $G'$ found using the approximated value of metric $\mathcal{MR}(D)$; i.e. $X=\max_{1\leq D\leq N} \min \{Z_D,D\}$. 

Suppose $Z_D^* \geq \frac{N}{k}$ for some constant $k>1$ \footnote{Note that for constant values of $D$ or $Z_D^*$, one can find an exact balanced biclique in polynomial time by an argument similar to Proposition \ref{Polynomial_MRD}. Therefore, the difficulty is for non-constant values of $D$ and $Z_D^*$; i.e. $D$ and $Z_D^*$ are a fraction of network. Thus, $k$ will a constant value.}. Then, the following equations hold.

\begin{subequations}
	\begin{align}
								 & W_D^* \leq W_D\leq r \cdot W_D^* \label{inapp_Eq1}\\
		\Rightarrow  & N-r \cdot W_D^* \leq N-W_D \leq N-W_D^* \\
		\Rightarrow  & -N(r-1)+r \cdot Z_D^* \leq Z_D \leq Z_D*  \\
		\Rightarrow  & [r-k(r-1)] \cdot Z_D^* \leq Z_D \leq Z_D*  \\
		\Rightarrow  & \min\{D, [r-k(r-1)] \cdot Z_D^*\} \leq \min\{D, Z_D\}  \nonumber \\
								 & \leq \min\{D, Z_D^* \}  \quad \forall D \in \{1\leq D\leq N\}\\
		\Rightarrow  & [r-k(r-1)] \cdot \min\{D, \cdot Z_D^*\} \leq \min\{D, Z_D\} \nonumber \\
								 & \leq \min\{D, Z_D^* \}  \quad \forall D \in \{1\leq D\leq N\}\\
		\Rightarrow  & [r-k(r-1)] \cdot \max_{1\leq D\leq N}\min\{D, \cdot Z_D\}  \nonumber \\
						     & \leq \max_{1\leq D\leq N}\min\{D, Z_D^*\} \leq \max_{1\leq D\leq N}\min\{D, Z_D^* \}  \\
		\Rightarrow  & [r-k(r-1)] \cdot X_D^* \leq X_D \leq X_D^*  \label{inapp_Eqn}
	\end{align}
\label{inapp_proof}
\end{subequations}

Equations (\ref{inapp_Eq1})-(\ref{inapp_Eqn}) indicate that for any $1 \leq r \leq 1+\frac{1-m\cdot.N^{-\epsilon}}{k-1}$ where $m$ and $\epsilon$ are some positive constants, if there exists an r-approximation for $\mathcal{MR}(D)$, there exists an $m\cdot N^{-\epsilon}$ approximation for the maximum balanced biclique (MBB) problem which is a contradiction \cite{feige2004hardness, khot2006ruling}.

\section{ILP formulation for 2-robust design}\label{Robust_ILP}

Here, we formulate the problem of allocating edges in a 2-robust bidirectional interdependent network. As discussed previously, by definition \ref{Lexi_Robust}, a 2-robust network is also 1-robust. Thus, we search for the optimal allocation in regular networks.

Let $N$ be the number of nodes in networks $A$ and $B$, and $k$ be the degree of each node. Let $X$ denote the lower bound on the number of neighbors of any pair of nodes; i.e. $X=\mathcal{MR}(D=2)$. Let $E \in \{0,1\}^{N\times N}$ be a binary matrix where $E_{ij}=1$ if node $i\in A$ is connected to node $j \in B$, and $E_{ij}=0$ otherwise. Moreover, let $Z_i^{jr}$ be a binary variable where $Z_i^{jr}=1$ if node $i \in A$ is a neighbor of node $j \in B$ or node $r \in B$, and $Z_i^{jr}=0$ otherwise. Furthermore, let $Y_j^{ir}$ be a binary variable where $Y_j^{ir}=1$ if node $j \in B$ is a neighbor of node $i \in A$ or node $r \in A$, and $Y_j^{ir}=0$ otherwise. The ILP formulation is as follows.

\begin{subequations}
\begin{align}
\max \quad & X \\
\mbox{s.t.}	\quad & \sum_{i=1}^{N} E_{ij} = k \quad j \in 1,\cdots,N            \label{degree_1}\\
			\quad & \sum_{j=1}^{N} E_{ij} = k \quad i \in 1,\cdots,N                   \label{degree_2}   \\
			\quad & Z_i^{jr} \leq E_{ij}+E_{ik} \quad j \in 1,\cdots,N-1;r \in j+1,\cdots,N \nonumber\\
			\quad & \quad \quad \quad \quad \quad \quad \quad \quad i\in 1,\cdots,N                \label{Pair_neighborsB}\\
			\quad & X \leq \sum_{i=1}^{N} Z_i^{jk} \quad j \in 1,\cdots,N-1;k \in j+1,\cdots,N      \label{bound_pairB} \\
			\quad & Y_j^{ir} \leq E_{ij}+E_{rj} \quad i \in 1,\cdots,N-1;r \in i+1,\cdots,N \nonumber\\
			\quad & \quad \quad \quad \quad \quad \quad \quad \quad j\in 1,\cdots,N                 \label{Pair_neighborsA}\\
			\quad & X \leq \sum_{j=1}^{N} Y_j^{ir} \quad i \in 1,\cdots,N-1;r \in i+1,\cdots,N      \label{bound_pairA}   \\		
			\quad & E_{ij} \in \{0,1\} \quad i \in 1,\cdots,N, j \in 1,\cdots,N      \label{binary1}\\
			\quad & Z_i^{jr} \in \{0,1\} \quad i \in 1,\cdots,N; j,r \in 1,\cdots,N               \label{binary2}\\
			\quad & Y_j^{ir} \in \{0,1\} \quad j \in 1,\cdots,N; i,r \in 1,\cdots,N                \label{binary3}
\end{align}
\label{ILP_2robust_simple}
\end{subequations}

The objective is to maximize the lower bound $X$; i.e. maximize metric $\mathcal{MR}(D=2)$. Constraints (\ref{degree_1}) and (\ref{degree_1}) guarantee that the degree of each node is $k$. Moreover, Constraints (\ref{Pair_neighborsB}) and (\ref{bound_pairB}) find the neighbors of any pair of nodes in network $B$ and set $X$ as the lower bound on the number of these neighbors. Similarly, Constraints (\ref{Pair_neighborsA}) and (\ref{bound_pairA}) find the neighbors of any pair of nodes in network $B$ and set $X$ as the lower bound on the number of these neighbors. Finally, constraints (\ref{binary1}-\ref{binary2}) sets the variables to be binary.

\end{document}